\documentclass[a4paper,twocolumn,11pt,accepted=2024-12-18]{quantumarticle}
\pdfoutput=1
\usepackage{graphicx}
\usepackage[dvipsnames]{xcolor}
\usepackage[numbers,sort&compress]{natbib}
\usepackage[ruled,linesnumbered]{algorithm2e}
\usepackage{amsmath,amsfonts,amssymb}
\usepackage{subfigure}
\usepackage{booktabs}
\usepackage{multirow}
\usepackage{todonotes}
\usepackage{array}
    \newcolumntype{P}[1]{>{\centering\arraybackslash}p{#1}}
    \newcolumntype{M}[1]{>{\centering\arraybackslash}m{#1}}
\AddToHook{begindocument/before}{\usepackage{hyperref}}
\usepackage[hyphens]{url}
\usepackage{ulem}
\usepackage{placeins}
\newtheorem{theorem}{Theorem}
\newtheorem{definition}{Definition}

\newtheorem{corollary}{Corollary}
\newcommand{\ms}{\mu}
\newcommand{\maxcut}{MAX-CUT }

\newcommand{\red}[1]{#1}

\newcommand{\cyan}[1]{}

\newcommand{\symbffrak}[1]{\mathfrak{#1}}
\def\qed{\hfill $\Box$}

\makeatletter
\newenvironment{proof}[1][\proofname]{\par
  \normalfont
  \topsep6\p@\@plus6\p@ \trivlist
  \item[\hskip\labelsep{\bfseries #1}\@addpunct{\bfseries.}]\ignorespaces
}{
  \endtrivlist
}
\makeatother

\title{Recursive Quantum Relaxation for Combinatorial Optimization Problems}

\author{Ruho Kondo}
\affiliation{Toyota Central R\&D Labs., Inc., 41-1, Yokomichi, Nagakute, Aichi 480-1192, Japan}
\affiliation{Quantum Computing Center, Keio University, 3-14-1 Hiyoshi, Kohoku-ku, Yokohama, Kanagawa 223-8522, Japan}
\email{r-kondo@mosk.tytlabs.co.jp}

\author{Yuki Sato}
\affiliation{Toyota Central R\&D Labs., Inc., 41-1, Yokomichi, Nagakute, Aichi 480-1192, Japan}
\affiliation{Quantum Computing Center, Keio University, 3-14-1 Hiyoshi, Kohoku-ku, Yokohama, Kanagawa 223-8522, Japan}

\author{Rudy Raymond}
\affiliation{Quantum Computing Center, Keio University, 3-14-1 Hiyoshi, Kohoku-ku, Yokohama, Kanagawa 223-8522, Japan}
\affiliation{Department of Computer Science, The University of Tokyo, 7-3-1, Hongo, Bunkyo-ku, Tokyo 113-0033, Japan}
\thanks{Former affiliation was IBM Quantum, IBM Japan.}

\author{Naoki Yamamoto}
\affiliation{Quantum Computing Center, Keio University, 3-14-1 Hiyoshi, Kohoku-ku, Yokohama, Kanagawa 223-8522, Japan}
\affiliation{Department of Applied Physics and Physico-Informatics, Keio University, Hiyoshi 3-14-1, Kohoku-ku, Yokohama 223-8522, Japan}

\begin{document}

\maketitle

\begin{abstract}
Quantum optimization methods use a continuous degree-of-freedom of quantum states to heuristically solve combinatorial problems, such as the \maxcut problem, which can be attributed to various NP-hard combinatorial problems. 
This paper shows that some existing quantum optimization methods can be unified into a solver to find the binary solution which is most likely measured from the optimal quantum state. 
Combining this finding with the concept of quantum random access codes (QRACs) for encoding bits into quantum states on fewer qubits, we propose an efficient recursive quantum relaxation method called recursive quantum random access optimization (RQRAO)  for MAX-CUT. 
Experiments on standard benchmark graphs with several hundred nodes in the \maxcut problem, conducted in a fully classical manner using a tensor network technique, show that RQRAO \red{not only} outperforms the Goemans–Williamson \red{and recursive QAOA methods, but also} is comparable to state-of-the-art classical solvers.
The code \red{is available at \url{https://github.com/ToyotaCRDL/rqrao}.}
\end{abstract}

\section{Introduction}
\label{sec:introduction}
\begin{figure*}[t]
    \centering
    \includegraphics[scale=1.2]{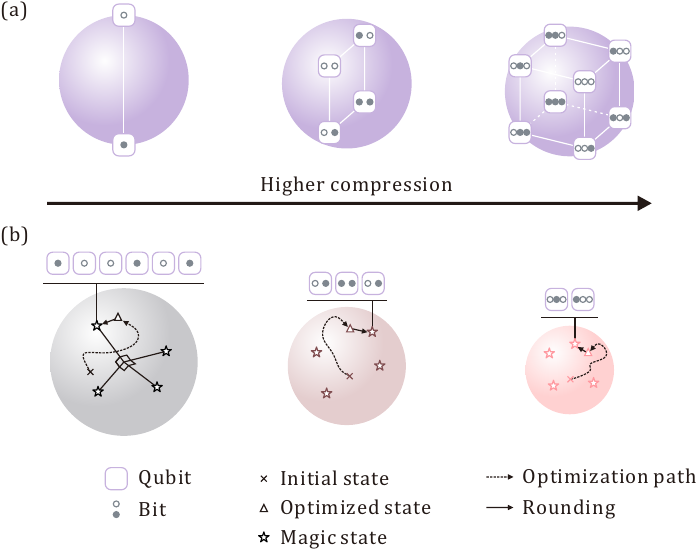}
    \caption{
    Concept of the proposed method, recursive quantum random access optimization (RQRAO).
    (a) One/two/three classical bits are embedded into one qubit using a quantum random access code (QRAC)~\cite{ambainis1999dense,ambainis2002dense}.
    A total of two/four/eight variants of the information of the one/two/three classical bits are represented as a vector pointing to the corner of the line/square/cube inscribed in the Bloch sphere.
    These specific quantum states are called (one-qubit) magic states.
    (b) Using a high compression QRAC, the size of the Hilbert space corresponding to the given problem is reduced and each (multi-qubit) magic state is no longer necessarily orthogonal to the others.
    Note that only four of the $2^6$ magic states are illustrated.
    After optimizing the objective function defined on the reduced Hilbert space, the bit string is decoded by searching for the nearest magic state as measured in the fidelity, which is called \textit{rounding}.
    To accomplish a good rounding, a recursive strategy similar to the RQAOA~\cite{bravyi2020obstacles} is adopted.
    }
  \label{fig:rqrao-concept}
\end{figure*}

The \maxcut problem is one of Karp's 21 NP-complete problems~\cite{karp2010reducibility} and is frequently used as a benchmark of combinatorial optimization algorithms because of its simplicity.
Namely, given a graph that consists of nodes connected by weighted edges, the optimization variant of the \maxcut problem asks for the labeling of nodes into two classes such that the sum of the weighted edges connecting nodes with different labels is maximized.
Obtaining good solutions for the \maxcut problem is in great demand, because various NP-hard combinatorial problems, e.g., partitioning, covering, and coloring problems, can be transformed into the \maxcut problem~\cite{lucas2014ising,glover2022quantum}.
Although the \maxcut problem is a discrete optimization problem, \red{it can be solved by relaxing the discrete variables into continuous ones.
For example, the semidefinite program in~\cite{goemans1995improved} finds an embedding of binary variables as real-valued vectors in a high-dimensional sphere so that the labels can be probabilistically determined by partitioning the vectors with a random hyperplane cutting.
(see Supplementary Information~\ref{sec:classical-maxcut}).
}
However, the quantum optimization does not take such an approximation because the probabilistic characteristics are naturally incorporated into the quantum states.
In \red{most approaches to the \maxcut problem using quantum algorithms,} all solutions are encoded in the eigenstates of the problem Hamiltonian, and then the quantum state is optimized such that the expected value of the problem Hamiltonian is maximized \red{(see Supplementary Information~\ref{sec:quantum-maxcut})}.
From a theoretical viewpoint, it has been proven that the eigenstate with the maximum eigenvalue can be obtained by the infinite-level quantum approximate optimization algorithm (QAOA)~\cite{qaoa1,wang2018quantum} with the appropriate parameters.
However, in a noisy quantum device, it is intractable to carry out even a sufficiently large level QAOA, let alone an infinite level QAOA.
For this reason, the field of combinatorial problems with quantum computing is dominated by the development of heuristics, that mainly focus on the development of the ansatz variants and improving the trainability of QAOA~\cite{blekos2023review}.

Rather than developing the QAOA ansatz, an alternative approach is to develop a non-diagonal Hamiltonian instead of the diagonal Ising-type one used in the QAOA.
In \cite{fuller2021approximate}, the \maxcut Hamiltonian is constructed using quantum random access codes (QRACs)~\cite{ambainis1999dense,ambainis2002dense} (see \red{Sec.~\ref{sec:qrao} for details, and Supplementary Information}~\ref{sec:qrac} for a review).
Using QRACs, the $n$-qubit quantum state can be used to encode $m$ binary variables for $m > n$.
The QRAC Hamiltonian is still a 2-local Hamiltonian, and therefore it is QMA-hard to find the ground state~\cite{kempe2006complexity}. However, the QRAC formulation has the advantage that the number of qubits can be reduced, making it easier to optimize the variational state.

Recently, \cite{bravyi2020obstacles} proposed the recursive QAOA (RQAOA), showing that the recursive application of level-1 QAOA gives an empirically higher approximation ratio than the na\"ive application of level-1 QAOA and the Goemans--Williamson (GW) algorithm~\cite{goemans1995improved}, where the latter is the best generic classical algorithm with a theoretical guarantee for the \maxcut problem under the unique games conjecture~\cite{khot2002power}.
This result is interesting because in \cite{bravyi2020obstacles}, it was proved that the approximation ratio of the na\"ive level-1 QAOA cannot outperform the GW for some specific graphs.
The application of the RQAOA has been extended to the MAX-$k$-CUT problem in \cite{bravyi2022hybrid}, which also shows that the RQAOA outperforms the best-known classical approximation algorithms.

Here, the following questions naturally arise: Is it possible to apply the QRAC in the RQAOA framework?
Furthermore, if it is possible, are there any benefits of using the QRAC in the RQAOA framework, such as improved approximation ratios?
In this paper, we propose a general framework for the \maxcut problem that incorporates the QRAC~\cite{fuller2021approximate} into the RQAOA~\cite{bravyi2020obstacles} framework, namely {\it recursive quantum random access optimization (RQRAO, see Fig.~\ref{fig:rqrao-concept})}.
Our contributions are summarized as follows:
\begin{enumerate}
    \item We theoretically show that incorporating QRACs into the RQAOA framework is a natural extension of \cite{fuller2021approximate} and \cite{bravyi2020obstacles}, i.e., they are all uniformly expressed as a method that (i) maximizes the expectation of the problem Hamiltonian so that all candidate solutions are embedded using QRACs, and (ii) searches for the nearest embedding as measured in the fidelity (Theorem~\ref{theorem:objective}).
    \item We give numerical experiments to show the benefit of incorporating QRACs into the RQAOA framework. Demonstrations on an open graph dataset, Gset~\cite{helmberg2000spectral}, were carried out using the tensor network simulation conducted on a classical computer.
    The results showed that the approximation ratio of the proposed method outperforms not only GW but also the methods from which inspiration was taken~\cite{fuller2021approximate,bravyi2020obstacles}, and is comparable to existing state-of-the-art classical heuristics~\cite{burer2002rank}. 
\end{enumerate}



Note that the proposed framework is entirely classical except for the ground state preparation step, and therefore it may not provide the so-called quantum advantage.
However, when an efficient quantum algorithm is used for the ground-state preparation of the framework, the overall runtime can be improved.

The remainder of the paper is organized as follows:
Section~\ref{sec:preliminaries} presents the mathematical definition of the \maxcut problem and describes the RQAOA and QRAC in brief.
Section~\ref{sec:proposed} describes the proposed algorithm and its theoretical background.
Section~\ref{sec:experiments} presents the results of numerical experiments.
Section~\ref{sec:conclusion} concludes the paper.

\section{Preliminaries}
\label{sec:preliminaries}

\subsection{\maxcut Problem}
\label{sec:maxcut}
As stated in Sec.~\ref{sec:introduction}, the goal of the \maxcut problem is to find binary node labels such that the sum of the weights of edges with different node labels is maximized.
It is mathematically formulated as follows:
\begin{definition}[\maxcut problem]
\label{def:maxcut}
    Given a graph $G=(V,E)$, where $V$ is a node set and $E$ is an edge set with edge weights $w_{jk}\in\mathbb{R}$ for $(j,k)\in E$, the \maxcut problem is defined as
    \begin{equation}
    \label{eq:def-maxcut}
    \max_{{\bf b}\in\{0,1\}^{|V|}} \texttt{CW}({\bf b})
    \end{equation}
    where
    \begin{equation}
    \label{eq:def-cut}
    \texttt{CW}({\bf b}):= \sum_{(j,k)\in E} w_{jk}\ \frac{1- (-1)^{b_j + b_k}}{2} 
    \end{equation}
    is the cut weight corresponding to ${\bf b} = (b_0, \ldots, b_{|V|-1})$, where $b_i \in \{0,1\}$ is an attribute of node $i$.
\end{definition}
\subsection{Recursive QAOA}
\label{sec:rqaoa}
Recursive QAOA (RQAOA)~\cite{bravyi2020obstacles} is a variant of QAOA for the \maxcut problem, which is an iterative method that determines the parity of one edge in a graph per iteration, deletes a node connected by the edge and modifies the graph.
Here, edge parity is defined as follows:
\begin{definition}[Edge parity]
    Given edge $(j, k)$, positive and negative edges are defined as $b_j=b_k$ and $b_j \neq b_k$, respectively. The edge parity is said positive (negative) for a positive (negative) edge. 
\end{definition}
In the first step of RQAOA, a candidate edge whose parity is to be determined is calculated using optimization.
Let $|\psi({\boldsymbol\theta})\rangle$ be a QAOA ansatz with learnable parameters $\boldsymbol\theta$ and
\begin{equation}
\label{eq:def-ising}
    H_{\rm Ising} := \sum_{(j,k)\in E} w_{jk}\,\frac{I-Z_jZ_k}{2} \\
\end{equation}
be the Ising-type \maxcut Hamiltonian, where $Z_j$ is a Pauli $Z$ matrix corresponding to the $j$th qubit.
The optimized parameters ${\boldsymbol \theta}^*$ are obtained by
maximizing $\langle \psi({\boldsymbol\theta}) | H_{\rm Ising} | \psi({\boldsymbol\theta})\rangle$. 
After optimization, the energy of edge $(j,k)\in E$ is calculated as
\begin{equation}
\label{eq:edge-energy-ZZ}
\mathcal{E}_{jk}
:=\langle\psi({\boldsymbol \theta}^*)|Z_jZ_k|\psi({\boldsymbol \theta}^*)\rangle.
\end{equation}
The candidate edge $(j^*,k^*)$ is defined as a maximal energy edge as
\begin{equation}
\label{eq:}
(j^*,k^*)=\mathop{\rm argmax}_{(j,k)\in E} |\mathcal{E}_{jk}|.
\end{equation}
Then, the parity of edge $(j^*,k^*)$ is determined as 
\begin{eqnarray}
\label{eq:parity}
\left\{
    \begin{array}{ll}
        b_{j^*}=b_{k^*} & \mathcal{E}_{j^*k^*}>0\\[5mm]
        b_{j^*}\neq b_{k^*} & \mathcal{E}_{j^*k^*}<0
    \end{array}
\right. .
\end{eqnarray}
Note that the probability that all edge energy is equal to zero ($\mathcal{E}_{j^*k^*}=0$) is negligibly small and can be ignored.

In the next step, the problem Hamiltonian is modified to satisfy Eq.~\eqref{eq:parity} as a constraint.
This modification can be interpreted as modifying the graph such that the number of its nodes is reduced by one, as shown in Fig.~\ref{fig:regraph}.
\begin{figure}[t]
    \centering
    \subfigure[]{
        \centering
        \includegraphics[scale=.85]{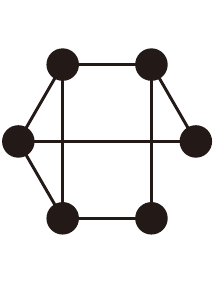}
        \label{fig:}
    }
    \subfigure[]{
        \centering
        \includegraphics[scale=.85]{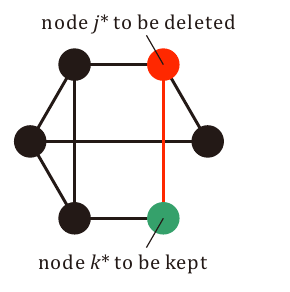}
        \label{fig:}
    }
    \\
    \centering
    \subfigure[]{
        \centering
        \includegraphics[scale=.85]{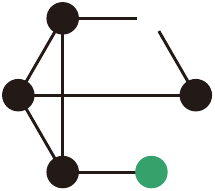}
        \label{fig:}
    }
    \subfigure[]{
        \centering
        \includegraphics[scale=.85]{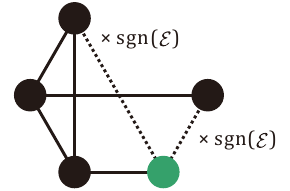}
        \label{fig:}
    }
    \caption{Graph modification procedure of RQAOA.
    (a) Initial graph.
    (b) The red edge indicates ${(j^*,k^*)=\mathop{\rm argmax}_{(j,k)}}|\mathcal{E}_{jk}|$, where $|\mathcal{E}_{jk}|:=|\langle\psi|Z_jZ_k|\psi\rangle|$.
    One of the nodes in edge $(j^*,k^*)$ is kept (the green node) while the other one is deleted (the red node).
    (c) After node removal.
    (d) Edges are re-connected to the retained node by multiplying ${\rm sgn}(\mathcal{E}_{j^*k^*})$ by its weight.}
  \label{fig:regraph}
\end{figure}
Node $j^*$ is deleted from the graph and edge $(j^*, l)$ $(l\neq k^*)$ is re-connected as $(k^*, l)$.
At the same time, edge weight $w_{k^*l}$ is updated to $w_{k^*l}+{\rm sgn}(\mathcal{E}_{j^*k^*})\cdot w_{j^*l}$ $(l\neq k^*)$.
As a result, one node is deleted from the graph by a single iteration.
Subsequently, the variational state is re-initialized and optimized using the Hamiltonian defined with respect to the modified graph.

The number of nodes is gradually decreased by repeating the edge parity determination and graph size reduction.
When the number of nodes becomes smaller than a threshold value $M$, the parities of the remaining nodes are determined by brute-force search.
\red{
Once all parities are determined, a label for one node is randomly set to either 0 or 1. 
Subsequently, the labels for all other nodes are determined based on the assigned parities.
}

\subsection{Quantum Random Access Optimization}
\label{sec:qrao}
In the Quantum Random Access Optimization (QRAO)~\cite{fuller2021approximate}, binary node attributes are encoded in fewer qubits than the number of nodes using quantum random access code (QRAC)~\cite{ambainis1999dense,ambainis2002dense}.
For each node $j\in V$, binary node attribute $b_j\in\{0,1\}$ is randomly mapped to $(q^{(j)}, P^{(j)})$, where $q^{(j)}$ is a qubit index and $P^{(j)}$ is a Pauli matrix with the following constraints:
\begin{equation}
\label{eq:constraint}
\left\{
    \begin{array}{ll}
        \displaystyle{ P^{(j)} \neq P^{(k)} } & {\rm for\ \ \ } q^{(j)}=q^{(k)}  \\[5mm]
        \displaystyle{ q^{(j)} \neq q^{(k)} } & {\rm \forall\ \ \ } (j,k)\in E
    \end{array}
\right. .
\end{equation}
In the quantum random access coding, there are variations that depend on the number of variants of the Pauli matrices.
If only $Z$ is available, $P^{(j)}=Z$, it reduces to the Ising-type formulation, which we refer to $(1,1)$-QRAC.
For $(2,1)$-QRAC, two variants of 1-local Pauli matrices, e.g., $P^{(j)}\in\{X,Z\}$, are available.
For $(3,1)$-QRAC, $P^{(j)}\in\{X,Y,Z\}$ are available.
Then, the relaxed \maxcut Hamiltonian, which we refer to as the $(m,1)$-QRAC Hamiltonian $(m=1,2,3)$, is defined as
\begin{equation}
\label{eq:h-qrac}
H_m := \sum_{(j,k)\in E} w_{jk}\,\frac{I-mP_{\langle j\rangle}P_{\langle k\rangle}}{2},
\end{equation}
where $P_{\langle j\rangle}:=P^{(j)}_{q^{(j)}}$ is a 1-local Pauli matrix corresponding to the $q^{(j)}$th qubit. 
If the given graph is sparse, the second constraint in  Eq.~\eqref{eq:constraint} rarely restricts the assignment, and hence, the total number of qubits can be reduced by up to $1/m$-th the number of nodes.

Let $\mu^{[q]}_m\in\mathbb{C}^{2\times 2}$ be the $q$th single qubit pure state in the form of a density matrix defined as
\begin{equation}
\label{eq:def-magic-state}
\mu^{[q]}_m(\{b^{[q]}_\bullet\})
:=\frac{1}{2}\Bigg(I+\frac{1}{\sqrt{m}} \sum_{P\in\mathfrak{P}}(-1)^{b^{[q]}_P}P\Bigg),
\end{equation}
where $\mathfrak{P}\subset\{X,Y,Z\}$, $|\mathfrak{P}|=m$, and $b^{[q]}_P\in\{0,1\}$ is a binary node attribute mapped to $(q, P)$.
\begin{figure}[t]
    \centering
    \includegraphics[scale=.8]{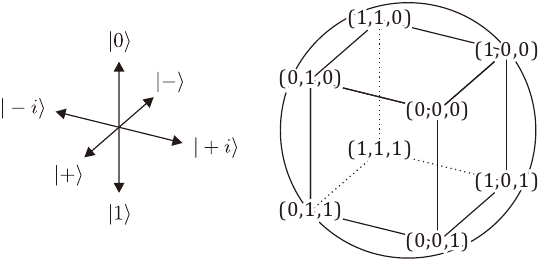}
    \caption{Representation of the $(3,1)$-magic state in the Bloch sphere. The coordinates of the vertex of the cube inscribed in the sphere correspond to the embedded classical bits $(b^{[q]}_X,b^{[q]}_Y,b^{[q]}_Z)$.}
  \label{fig:bloch}
\end{figure}
Because $\mu^{[q]}_m(\{b^{[q]}_\bullet\})$ is a one-qubit pure state, it can be drawn in the Bloch sphere (Fig.~\ref{fig:bloch}).
For example, in the case of $m=3$, the states $\mu^{[q]}_3(\{b^{[q]}_\bullet\})$ for $(b^{[q]}_X,b^{[q]}_Y,b^{[q]}_Z)\in\{0,1\}^3$ correspond to the vectors pointing the vertices of the cube inscribed in the Bloch sphere.
The $n$-qubit product state of $\mu^{[q]}_m$ is defined as
\begin{equation}
\label{eq:product-state}
\mu_m({\bf b})
    :=\mu_m^{[n-1]}(\{b^{[n-1]}_\bullet\}) \otimes \cdots \otimes \mu_m^{[0]}(\{b^{[0]}_\bullet\}).
\end{equation}
According to \cite[Proposition 1]{fuller2021approximate}, the following property holds.
\begin{equation}
\label{eq:}
    {\rm tr}(H_m\,\mu_m({\bf b})) = \texttt{CW}({\bf b}).
\end{equation}
Of course, the $n$-qubit quantum state is not required to be a product state such as Eq.~\eqref{eq:product-state}.
If $\rho$ is assumed to be an arbitrary $n$-qubit state that includes the aforementioned product state, ${\rm tr}(H_m\,\rho)$ could be higher than ${\rm tr}(H_m\,\mu({\bf b}))$ when $m>1$, where $\rho$ may not be directly related to an optimal binary solution.
This is why $H_m$ $(m>1)$ is called the \textit{relaxed} Hamiltonian.
Note that it has not been clarified what exactly ${\rm tr}(H_m\,\rho)$ represents when $\rho$ is an arbitrary state in \cite{fuller2021approximate}, but this will be clarified in Corollary~\ref{corollary:qrac-hamiltonian} in Sec.~\ref{sec:theoretical}. 

After maximizing ${\rm tr}(H_m\rho)$ with respect to $\rho$, the classical bit string ${\bf b}$ is decoded from the optimized $\rho$.
This process is called \red{{\it rounding}}.
In \cite{fuller2021approximate}, two kinds of rounding were proposed, magic state rounding and Pauli rounding. 
Note that magic state rounding is the projective measurement whose basis is the magic state selected uniformly at random.
The detailed protocol is defined in Definition~\ref{def:magic-meas} of Supplementary Information~\ref{sec:shadow-magic}, and as can be seen to closely resemble the classical shadow with Random Pauli measurements~\cite{huang2020predicting}.
For this reason, the magic state rounding is henceforth referred to as {\it random magic measurement}. 
Note that $(1,1)$-random magic measurement with $\mathfrak{P}=\{Z\}$ is equivalent to measurements on a computational basis.
Pauli rounding is a decoding method based on the value of ${\rm tr}(P_{\langle j \rangle}\rho)$.
The value of ${\bf b}$ is determined as
\red{
\begin{equation}
\label{eq:pauli-rounding-org}
b_{j}=\left\{
\begin{array}{ll}
0 & \big({\rm tr}(P_{\langle j \rangle}\rho)>0\big)\\[2mm]
1 & \big({\rm tr}(P_{\langle j \rangle}\rho)<0\big)
\end{array}
\right..
\end{equation}
}
If ${\rm tr}(P_{\langle j \rangle}\rho)=0$, $b_j$ is determined as $0$ or $1$ uniformly at random.
In Sec.~\ref{sec:theoretical}, we provide theorems to show the relationship between random magic measurements and Pauli rounding and reveal what assumption is made in Pauli rounding.

\section{Proposed Method}
\label{sec:proposed}
The proposed method, named RQRAO, is introduced in this section.

Section~\ref{sec:algorithm} describes the concrete algorithm of the proposed method.
In brief, the proposed method incorporates the $(m,1)$-QRAC Hamiltonian (Eq.~\eqref{eq:h-qrac}) into the RQAOA framework (Sec.~\ref{sec:rqaoa}).
The edge parity is recursively determined based on the value of the edge energy, which is defined using the $(m,1)$-QRAC Hamiltonian.
Two new heuristics are also incorporated into the proposed method: the ensemble method of defining the edge energy, and determining several edge parities at once using the maximum spanning tree.
The former improves the resulting cut weight whereas the latter reduces the runtime of the proposed method.

In Sec.~\ref{sec:theoretical}, we show that the proposed method, RQAOA, and QRAO 
\begin{enumerate}
    \item maximize the expected cut weight using $(m,1)$-random magic measurements with respect to the quantum state, and
    \item decode the bit string embedded in the nearest magic state as measured in the fidelity against the optimized quantum state.
\end{enumerate}

\subsection{Algorithm}
\label{sec:algorithm}

The pseudocode of RQRAO is presented in Algorithm~\ref{alg:recursive}.
The algorithm describes how the \textit{ensemble edge energy} $\mathfrak{E}_{jk}$ of edge $(j,k)$ is calculated and how the edge parity is calculated from $\mathfrak{E}_{jk}$.
The details are described below.

\begin{algorithm*}[tb]
    \caption{Algorithm for recursive quantum random access optimization (RQRAO).}
    \label{alg:recursive}
    \SetKwInOut{Input}{input}
    \SetKwInOut{Output}{output}
    \SetKwInOut{Initialize}{initialize}
    \SetKwRepeat{Do}{do}{while}
    \Input{Graph $G=(V,E)$, number of ensembles $N$, scale factor $S$, brute-force search threshold $M$}
    \Output{\maxcut solution ${\bf b}^{*}$}
    \Initialize{$\mathcal{P}\leftarrow \emptyset$, $G'=(V',E')\leftarrow G=(V,E)$}
    $w'_{jk} \leftarrow w_{jk} + \xi,\ \ \ \xi\sim{\rm Uniform}([-10^{-5},10^{-5}])$\;
    \While{$|V'|>M$}{
        \tcc{Compute the ensemble edge energy}
        \For{$t=1,\ldots,N$}{
            Assign Pauli $P^{(t)}$ to each node $j\in V'$ and make the $(m,1)$-QRAC Hamiltonian $H_m^{(t)}(G')$\;
            $\theta^{(t)*} \leftarrow $ Result of maximization of $\langle \psi(\theta)| H_m^{(t)}(G') |\psi(\theta)\rangle$\;
            $\displaystyle{\mathcal{E}^{(t)}_{jk} \leftarrow \langle \psi(\theta^{(t)*})| P^{(t)}_{\langle j \rangle} P^{(t)}_{\langle k \rangle} ~ |\psi(\theta^{(t)*})\rangle}$ for all $(j,k)\in E'$\;
        }
        \For{$(j,k)\in E'$}{
            $\mu_{jk}\leftarrow \texttt{Mean}(\{\mathcal{E}_{jk}^{(t)}\})$\;
            $\sigma_{jk}\leftarrow \texttt{StandardDeviation}(\{\mathcal{E}_{jk}^{(t)}\})$\;
            $\displaystyle{\mathfrak{E}_{jk} \leftarrow} \mu_{jk}- {\rm sign}(\mu_{jk}) \cdot {\rm min}(S \sigma_{jk}, |\mu_{jk}|)$\;
        }
        \tcc{Compute the edge parity and modify the graph} 
        {{
            $E_{\rm S} \leftarrow \{ (j,k,|\mathfrak{E}_{jk}|):(j,k)\in E', |\mathfrak{E}_{jk}|>0\}$\;
            $E_{\rm T}\leftarrow\texttt{MaximumSpanningTree}(E_{\rm S})$\;
            \For{$(j,k)\in E_{\rm T}$}{
                $\mathcal{P}\leftarrow \mathcal{P}\cup\{(j,k,{\rm sign}(\mathfrak{E}_{jk}))\}$\;
                $G'=(V',E') \leftarrow \texttt{GetReducedGraph}(G', (j,k))$
            }
        }}
    }
    ${\bf b}^* \leftarrow \texttt{BruteForceSearch}(G, \mathcal{P})$\;
\end{algorithm*}

\paragraph{Line 1: Add Noise to Edges}

Add small perturbation noise to all edge weights to mitigate the \textit{isolated nodes problem} (see Supplementary Information~\ref{sec:nodes-isolation}).
We note that the isolated nodes problem also occurs at RQAOA~\cite{bravyi2020obstacles} although not reported in the previous literature.
We found out the perturbation is effective to deal with the problem. 

\paragraph{Line 4: Define the Relaxed Hamiltonian}
Assign 1-local Pauli $P^{(t)}$ to each node included in the current graph $G'=(V',E')$, where superscript $(t)$ describes the $t$th trial.
Note that $t$ is the sample index of the ensemble, not to be confused with the recursive iteration index.

The assignment is randomly carried out so as not to violate constraints Eq.~\eqref{eq:constraint}.
Then, the $t$th relaxed Hamiltonian is defined as
\begin{equation}
\label{eq:}
H_m^{(t)}:=\sum_{(j,k)\in E'} \frac{I-mP^{(t)}_{\langle j\rangle} P^{(t)}_{\langle k\rangle}}{2}w'_{jk} .
\end{equation}

\paragraph{Line 5: Optimize the Quantum State}

Define ansatz $|\psi({\boldsymbol\theta})\rangle$, where ${\boldsymbol\theta}\in\mathbb{R}^K$ are $K$ learnable parameters.
\red{For example, we used the matrix product state ansatz in the numerical experiments, where all matrix elements were treated as the learnable parameters (see Implementation Details in Sec.~\ref{sec:experiments}).}
The objective function is defined as
\begin{equation}
\label{eq:}
\mathcal{L}^{(t)}({\boldsymbol\theta}):=\langle{\psi}(\boldsymbol\theta)|H_m^{(t)}|{\psi}(\boldsymbol\theta)\rangle.
\end{equation}
The optimized quantum state $|\psi({\boldsymbol\theta}^*)\rangle$ is obtained by maximizing $\mathcal{L}^{(t)}({\boldsymbol\theta})$.
Any optimization method for the parameters of the ansatz can be used here.

\paragraph{Line 6: Compute the Edge Energy}
The $t$th edge energy is computed as 
\begin{equation}
\label{eq:def-energy-t}
\mathcal{E}^{(t)}_{jk}
:=\langle\psi({\boldsymbol\theta}^*)|P^{(t)}_{\langle j \rangle}P^{(t)}_{\langle k \rangle}|\psi({\boldsymbol\theta}^*)\rangle.
\end{equation}
The resulting edge energy strongly depends on both the assignment of 1-local Pauli $P^{(t)}$ and the reached quantum state, which is in general not a eigenstate corresponding to the maximum eigenvalue of $H_m^{(t)}$.

\paragraph{Lines \red{8--12}: Compute the Ensemble Edge Energy}

Using $N$ edge energies, $\{\mathcal{E}^{(t)}_{jk}\}^N_{t=1}$, the ensemble edge energy is calculated as 
\begin{equation}
\label{eq:def-energy}
\mathfrak{E}_{jk}
:=\mu_{jk}- {\rm sign}(\mu_{jk}) \cdot {\rm min}(S \sigma_{jk}, |\mu_{jk}|)
\end{equation}
where $\mu_{jk}$ and $\sigma_{jk}$ are the mean and standard deviation of $\{\mathcal{E}^{(t)}_{jk}\}^N_{t=1}$, which are defined as $\mu_{jk} := \frac{1}{N} \sum^N_{t=1}\mathcal{E}^{(t)}_{jk}$ and $\sigma_{jk} := \sqrt{\frac{1}{N} \sum^N_{t=1}(\mathcal{E}^{(t)}_{jk}-\mu_{jk})^2}$, respectively
\red{, and $S$ is a hyperparameter designed such that if $|\mu_{jk}| \leq S\sigma_{jk}$, then $\mathfrak{E}_{jk} = 0$.}

\red{
The reason for taking an ensemble is twofold.
First, unlike QAOA whose edge energy is solely from the ZZ interaction as in Eq.~\eqref{eq:edge-energy-ZZ}, the edge energy of QRAO can be defined by 9 different pairs of 1-local Pauli as in Eq.~\eqref{eq:def-energy-t} whose corresponding Hamiltonian and ground state can be significantly different: some may be easier to compute with certain ansatz than others.
The ensemble is utilized to take advantage of the various Hamiltonians and ground states.
Second, taking an ensemble average of the edge energy from different pairs of 1-local Pauli is akin to \textit{bagging}~\cite{breiman1996bagging} in machine learning.
Namely, even if an individual pair of 1-local Pauli is a weak predictor of the edge parity, their averages can give higher predictive performance when they are weakly correlated with each other. 
In our method, when optimization is performed on each of the $N$ different $(m,1)$-QRAC Hamiltonians created from different Pauli assignments, $N$ quantum states are obtained.
Since these quantum states are optimized with respect to different Hamiltonians as objective functions, they can be considered analogous to the weak predictors mentioned above.
Therefore, we expect an ensemble edge energy from these quantum states will yield a stronger predictor of the edge parity.
}

\paragraph{Line \red{14}: Compute the Maximum Spanning Tree}
After computing the ensemble edge energy $\mathfrak{E}_{jk}$, the parities of edges with nonzero $\mathfrak{E}_{jk}$ are determined according to the sign of $\mathfrak{E}_{jk}$.
However, sometimes inconsistency arises.
For example, this occurs when three negative edges form a cycle $\{(b_0, b_1), (b_1, b_2), (b_2, b_0)\}$, one of which is inconsistent because $(b_0\neq b_1)\land(b_1\neq b_2)\land(b_2\neq b_0)$ cannot hold.
To avoid this inconsistency, we adopt the maximum spanning tree to remove any cycles from the candidates.

Edges with nonzero $\mathfrak{E}_{jk}$ are collected and $|\mathfrak{E}_{jk}|$ are treated as edge weights.
For each subgraph $E_{\rm S}$ made by this collection that is disconnected from the others, the maximum spanning tree $E_{\rm T}$ can be calculated in linear time in the number of edges by Prim's algorithm~\cite{Prim1957} or by the randomized KKT algorithm~\cite{KargerKleinTarjan1995}. 
The obtained tree $E_{\rm T}$ can be easily converted into a rooted tree by setting an arbitrary node as the root.

\paragraph{Lines \red{15--18}: Modify the Graph}
Now, we have candidates $E_{\rm T}$ that can be used to determine the parity and can be deleted from graph $G'$ without inconsistency.
The parity for each edge in $E_{\rm T}$ is determined from the leaf to the root one by one according to its sign.
At the same time, the node on the leaf side is deleted and the edge weights are updated.
The graph modification procedure is the same as the RQAOA, as depicted in Fig.~\ref{fig:regraph}.

\begin{figure}[t]
    \centering
    \subfigure[]{
        \centering
        \includegraphics[scale=.85]{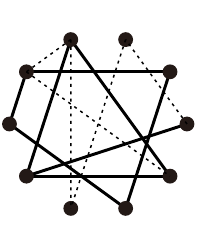}
        \label{fig:trees-a}
    }
    \subfigure[]{
        \centering
        \includegraphics[scale=.85]{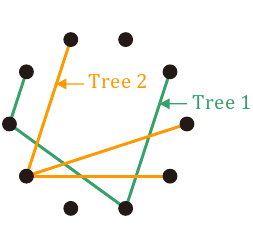}
        \label{fig:trees-b}
    }
    \subfigure[]{
        \centering
        \includegraphics[scale=.85]{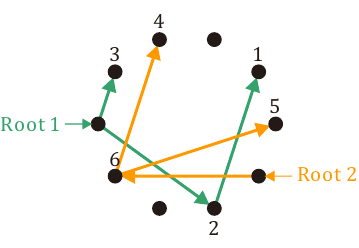}
        \label{fig:trees-c}
    }
    \caption{\red{Schematic illustration of tree construction and node deletion ordering: (a) Solid and dashed lines indicate edges with energy above and below the threshold, respectively; (b) Trees formed from the solid lines in (a); (c) Root nodes are randomly selected for each tree, and the nodes to be deleted are ordered from the leaf to the root.}}
  \label{fig:trees}
\end{figure}

\red{
Here, we note that the proposed method differs from the RQAOA in two parts: on which nodes and on how many nodes to be removed per iteration. 
In RQAOA, only one node at one end of the edge with the maximum absolute edge energy is removed per iteration, which is essentialy the only choice up to symmetricity of the ZZ interaction. 
In contrast, in the proposed method, any node that is an endpoint of an edge whose absolute edge energy exceeding the threshold is a candidate for removal per iteration (Fig.~\ref{fig:trees-a}).
To perform this multiple node removal without inconsistency, some edges that form a cycle are removed by calculating a spanning tree.
If an edge is not part of the spanning tree, it is ignored, even if its edge energy exceeds the threshold.
In the step where spanning trees are formed from the edges exceeding the threshold, a tree is created for each disconnected set of edges (Fig.~\ref{fig:trees-b}).
For each tree, a rooted tree is created by selecting a node at random as the root (Fig.~\ref{fig:trees-c}), and node removal is performed from the leaves toward the root.
}

\paragraph{Line \red{20}: Brute-force Search}
This is essentialy the same as in RQAOA~\cite{bravyi2020obstacles}:  the recursive procedure is stopped when the number of nodes of $G'$ is less than or equal to some constant $M$. In such case, the number of nodes whose parities are undetermined is at most $M$, and therefore the number of candidate solutions is at most $2^M$. The best solution ${\bf b}^*$ can be quickly obtained by a brute-force search. 

\paragraph{Example}

\begin{figure*}[t]
    \centering
    \subfigure[]{
        \centering
        \includegraphics[scale=.8]{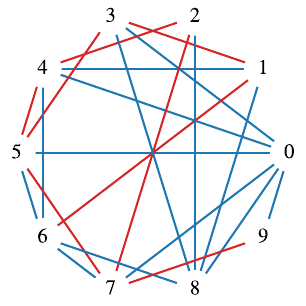}
        \label{fig:rqrao-example-a}
    }
    \subfigure[]{
        \centering
        \includegraphics[scale=.8]{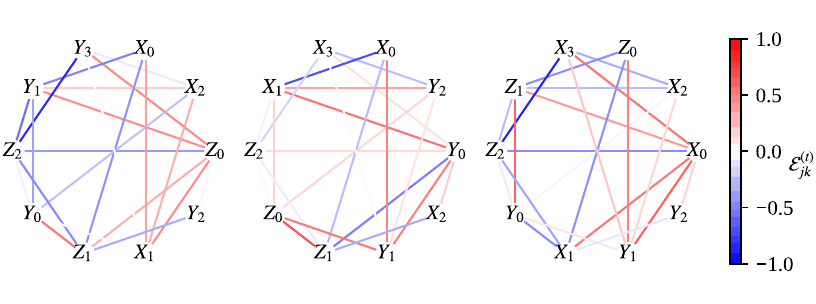}
        \label{fig:rqrao-example-b}
    }
    \subfigure[]{
        \centering
        \includegraphics[scale=.8]{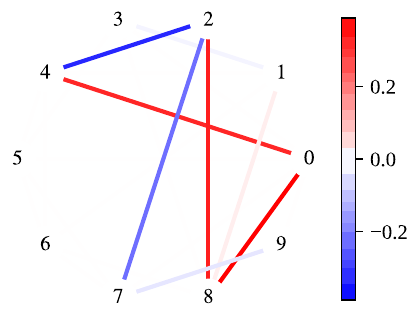}
        \label{fig:rqrao-example-c}
    }
    \subfigure[]{
        \centering
        \includegraphics[scale=.8]{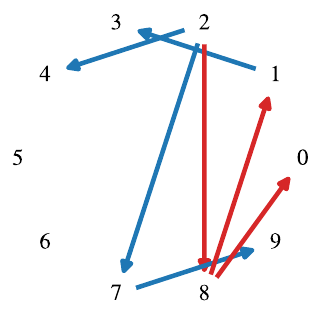}
        \label{fig:rqrao-example-d}
    }
    \subfigure[]{
        \centering
        \includegraphics[scale=.8]{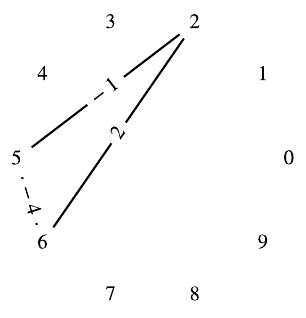}
        \label{fig:rqrao-example-e}
    }
    \caption{
    Schematic illustration of the RQRAO procedure. (a) Initial graph whose edge weights are $+1$ (red) or $-1$ (blue).
    (b) Three samples of edge energy, $\mathcal{E}^{(1)}_{jk}$, $\mathcal{E}^{(2)}_{jk}$, and $\mathcal{E}^{(3)}_{jk}$, are evaluated using optimized states whose corresponding Hamiltonians are $H^{(1)}_3$, $H^{(2)}_3$, and $H^{(3)}_3$, respectively. Node labels are assigned 1-local Paulis.
    (c) Ensemble edge energy $\mathfrak{E}_{jk}$ using the three samples illustrated in (b). Closed loop $0-4-2-8-0$ is inconsistent because $(b_0=b_4)\land(b_4\neq b_2)\land(b_2=b_8)\land(b_8=b_0)$ cannot be satisfied.
    (d) Maximum spanning tree with $|\mathfrak{E}_{jk}|$ edge weights. Red and blue indicate positive and negative edges, respectively. Edge $(0,4)$ is ignored at this step.
    (e) Modified graph after all leaf nodes of (d) have been removed and the edges in the tree have been reconnected to the remaining nodes. RQRAO repeats steps (a) to (e) until the number of nodes becomes smaller than $M$.
    }
  \label{fig:rqrao-example}
\end{figure*}

The whole RQRAO procedure is schematically shown in Fig.~\ref{fig:rqrao-example}.
The problem graph $G$ is shown in Fig.~\ref{fig:rqrao-example-a}, where the red and blue edges have weights $+1$ and $-1$, respectively.
The random Pauli assignments and optimization are carried out to obtain the edge energy for each trial $t$, $\mathcal{E}^{(t)}_{jk}$ (Fig.~\ref{fig:rqrao-example-b}).
The ensemble edge energy $\mathfrak{E}_{jk}$ is then calculated using $\{\mathcal{E}^{(t)}_{jk}\}^N_{t=1}$ (Fig.~\ref{fig:rqrao-example-c}).
The inconsistent edges are removed by solving the maximum spanning tree problem, where $|\mathfrak{E}_{jk}|$ are regarded as the edge weights.
The rooted tree is constructed by giving the direction from root to leaf, where the root node has been randomly selected (Fig.~\ref{fig:rqrao-example-d}). 
The edge parities in the rooted tree are determined from leaf to root according to the sign of corresponding $\mathfrak{E}_{jk}$.
At the same time, the graph is modified to satisfy the determined parity.
Figure~\ref{fig:rqrao-example-e} represents the modified graph after all edge parities in the rooted tree have been determined.
The protocol from Fig.~\ref{fig:rqrao-example-a} to \ref{fig:rqrao-example-e} is repeated until the number of nodes of the modified graph becomes at most $M$.

\subsection{Theoretical Analysis}
\label{sec:theoretical}

The proposed method, RQRAO, combines the existing RQAOA~\cite{bravyi2020obstacles} and QRAO~\cite{fuller2021approximate}.
Although the existing methods work well experimentally, there are interesting theoretical questions:
(i) Why does the recursive approach improve the cut weight performance in the RQAOA?
(ii) What is the meaning of maximizing the expectation value of the QRAC Hamiltonian, and what is the theoretical background of Pauli rounding in QRAO?
To answer these questions, we prove the following theorem:
\begin{theorem}
\label{theorem:objective}
The objective of the RQAOA, QRAO, and RQRAO can be unifiedly expressed as the problem of solving
\begin{equation}
\label{eq:prop2}
    \rho^*=\mathop{\rm argmax}_{\rho} \mathbb{E}_{{\bf b}\sim \mathcal{P}_m({\bf b};\rho)} \Big[\texttt{CW}({\bf b})\Big],
\end{equation}
\begin{equation}
\label{eq:prop1}
    {\bf b}^{\rm R}=\mathop{\rm argmax}_{{\bf b}\in\{0,1\}^{|V|}} \mathcal{P}_m({\bf b};\rho^*)\red{+\mathcal{P}_m(\overline{{\bf b}};\rho^*)},
\end{equation}
where \red{$\overline{\bf b}$ is the bit-flip of ${\bf b}$, and} $\mathcal{P}_m({\bf b};\rho)$ is the probability of obtaining ${\bf b}$ by $(m,1)$-random magic measurements whose value is proportional to the fidelity between $\rho$ and the $(m,1)$-magic state embedding $\mu({\bf b})$.
\end{theorem}

\red{Here, the superscript ``R'' stands for ``rounding'' and emphasizes that it represents the classical bits obtained by decoding the information embedded in the qubits.}
From Theorem~\ref{theorem:objective}, it is clarified that both the RQAOA and QRAO, let alone RQRAO, are models searching for the highest probability sample of the probability distribution that maximizes the expectation value of the cut weight through the quantum state as the hidden variable.
It will be proved later that maximizing the expectation value of the QRAC Hamiltonian, let alone that of the Ising-type Hamiltonian, can be seen as carrying out Eq.~\eqref{eq:prop2}.
In addition, it will also be proved later that both the recursive approach in the RQAOA and Pauli rounding in the QRAO are just methods that approximate Eq.~\eqref{eq:prop1}.
\red{
One of the contributions of this paper is a discovery that QRAO, RQAOA, and RQRAO, which at first glance seem unrelated in terms of their decoding processes, all decode bit strings according to Eq.~\eqref{eq:prop1}.
This allows us to consider a unified framework that encompasses all methods.
}

Surely, Theorem~\ref{theorem:objective} does not say the \maxcut problem is efficiently solvable because both maximizing the expected cut weight and obtaining the highest probability sample is, in general, intractable when ${\bf b}$ lives in an extremely high dimensional space.
However, we can say that it is natural to take the ${\bf b}$ with the highest probability because, after solving Eq.~\eqref{eq:prop2}, which maximizes the expected cut weight, the probability for larger $\texttt{CW}({\bf b})$ is expected to be large.
It is therefore reasonable to solve Eqs.~\eqref{eq:prop2} and \eqref{eq:prop1} to obtain a better cut weight.

In Sec.~\ref{sec:prop2}, Eq.~\eqref{eq:prop2} is proved, while Eq.~\eqref{eq:prop1} is proved in Sec.~\ref{sec:prop1}.
All the proofs of the theorems and corollaries in the following section are deferred to Supplementary Information~\ref{sec:proof}.

\subsubsection{Proof of Eq.~\eqref{eq:prop2}}
\label{sec:prop2}

The objectives of the RQAOA, QRAO, and RQRAO are to maximize ${\rm tr}(H_m\rho)$, where $m=1$ for the RQAOA whereas $m\in\{2,3\}$ for QRAO and RQRAO.
Hence, to prove Eq.~\eqref{eq:prop2}, we need to show that maximizing ${\rm tr}(H_m\rho)$ is equivalent to maximizing $\mathbb{E}_{{\bf b}\sim \mathcal{P}_m({\bf b};\rho)} \big[\texttt{CW}({\bf b})\big]$.
For this, let us start with the following theorems and corollaries, which are given for the first time in this paper. 

\begin{theorem}
\label{theorem:fidelity}
Let $P^{\otimes k}:=\prod^k_{i=1}P_{q^{(i)}}\in{\mathbb C}^{2^n\times 2^n}$ be the $k$-local Pauli observable, where $P_q$ represents $1$-local Pauli $P\in\mathfrak{P}\subset\{X,Y,Z\}$ acting on the $q$th qubit, and $q^{(i)}$ is the $i$th qubit index, where $q^{(i)}\neq q^{(j)}$ for $i\neq j$.
Suppose the case $|\mathfrak{P}|=m$.
Let $\rho\in{\mathbb C}^{2^n\times 2^n}$ be the arbitrary $n$ qubit state, $\mu_m({\bf b})\in{\mathbb C}^{2^n\times 2^n}$ be the $n$ qubit $(m,1)$-magic state defined in Eq.~\eqref{eq:product-state}, in which ${\bf b}\in\{0,1\}^{mn}$ is embedded, and $f_m({\bf b},\rho):={\rm tr}(\mu_m({\bf b})\rho)$ be the fidelity between $\mu_m({\bf b})$ and $\rho$.
The following relationship holds:
\begin{equation}
\label{eq:k-local-exp}
{\rm tr}\big( P^{\otimes k} \rho \big)
= m^{k/2} \, \mathbb{E}_{{\bf b}\sim \mathcal{P}_m({\bf b};\rho)}\Big[(-1)^{\sum^k_{i=1} b^{[q^{(i)}]}_{P}}\Big],
\end{equation}
where $b^{[q]}_{P}\in\{0,1\}$ is a binary mapped to the Pauli $P$ of the $q$th qubit and
\begin{equation}
\label{eq:pf}
\mathcal{P}_m({\bf b};\rho) := \frac{f_m({\bf b},\rho)}{\sum_{{\bf b}'\in\{0,1\}^{mn}}f_m({\bf b}',\rho)}.
\end{equation}
\end{theorem}

Note that ${\bf b}$ following the probability $\mathcal{P}_m$ is the outcome of the $(m,1)$-random magic measurements, whose details are described in Supplementary Information~\ref{sec:shadow-magic}.

\begin{corollary}
\label{corollary:qrac-hamiltonian}
Let $\rho\in{\mathbb C}^{2^n\times 2^n}$ be the arbitrary $n$ qubit state and $H_m$ be the $(m,1)$-QRAC Hamiltonian defined as Eq.~\eqref{eq:h-qrac}.
The following relationship holds:
\[\displaystyle{{\rm tr}(H_m \, \rho)}\]
\[\displaystyle{
 = \sum_{(j,k)\in E} \frac{1-m^2\,\mathbb{E}_{{\bf b}\sim \mathcal{P}_m({\bf b};\rho)}\big[(-1)^{b_j+b_k}\big]}{2} w_{jk}
}\]
\begin{equation}
\label{eq:tr-qrac}
 = m^2\,\mathbb{E}_{{\bf b}\sim \mathcal{P}_m({\bf b};\rho)}\big[\texttt{CW}({\bf b})\big]
-\frac{m^2-1}{2}\sum_{(j,k)\in E}w_{jk},
\end{equation}
where $\mathcal{P}_m$ is defined as Eq.~\eqref{eq:pf}.
\end{corollary}

According to Eq.~\eqref{eq:tr-qrac}, the expectation value of the $(m,1)$-QRAC Hamiltonian is proportional to the expectation value of the cut weight, where the expectation is over the bit strings obtained from $(m,1)$-random magic measurements.
Hence, maximizing the expectation value of the $(m,1)$-QRAC Hamiltonian can be regarded as maximizing the expected cut weight.
Note that the lower bounds of $\mathbb{E}_{{\bf b}\sim \mathcal{P}_m({\bf b};\rho)}\big[\texttt{CW}({\bf b})\big]$ for $w_{jk}=1$ when ${\rm tr}(H_m\rho)\ge\texttt{CW}({\bf b}^*)$ have been given in \cite[Theorem 2 and Supplementary Information VII]{fuller2021approximate}.
They are $\mathbb{E}_{{\bf b}\sim \mathcal{P}_m({\bf b};\rho)}\big[\texttt{CW}({\bf b})\big] \ge \frac{5}{9}\texttt{CW}({\bf b}^*)$ for $m=3$ and $\mathbb{E}_{{\bf b}\sim \mathcal{P}_m({\bf b};\rho)}\big[\texttt{CW}({\bf b})\big] \ge \frac{5}{8}\texttt{CW}({\bf b}^*)$ for $m=2$.
These bounds are reproduced by substituting ${\rm tr}(H_m\rho)\ge\texttt{CW}({\bf b}^*)$ and $\sum_{(j,k)\in E}w_{jk}\ge\texttt{CW}({\bf b}^*)$ into Eq.~\eqref{eq:tr-qrac}.

Note that not only the expectation value but also the variance of the cut weight can be obtained using the $(m,1)$-QRAC Hamiltonian as follows.
\begin{corollary}
\label{corollary:qrac-variance}
Let $\rho\in{\mathbb C}^{2^n\times 2^n}$ be the arbitrary $n$ qubit state and $H_m$ be the $(m,1)$-QRAC Hamiltonian defined as Eq.~\eqref{eq:h-qrac}.
The variance of the cut weight is
\[\displaystyle{
{\rm Var}_{{\bf b}\sim \mathcal{P}_m({\bf b};\rho)}\big[\texttt{CW}({\bf b})\big]
}\]
\begin{equation}
\label{eq:qrac-variance}
=\frac{1}{m^4}\Big({\rm tr}(H_m^2\rho)-{\rm tr}(H_m\rho)^2\Big),
\end{equation}
where $\mathcal{P}_m$ is defined as Eq.~\eqref{eq:pf}.
\end{corollary}

\subsubsection{Proof of Eq.~\eqref{eq:prop1}}
\label{sec:prop1}

Because the cut weight is invariant to the bit flip, we can fix one of the \red{$b_{j}$} to $0$ or $1$.
In this case, there exists a one-to-one relationship between $\mathcal{P}_m({\bf b};\rho^*)\red{+\mathcal{P}_m(\overline{\bf b};\rho^*)}$ and
\begin{equation}
\label{eq:hypo-general-parity}
{\mathcal{P}}_{m}(\symbffrak{p};\rho^*,T)= 
\prod^{|V|-1}_{b=1} \mathcal{P}_m(\mathfrak{p}_{e_b}|\symbffrak{p}_{e_{<b}};\rho^*,T),
\end{equation}
where 
\begin{equation}
\mathfrak{p}_{e=(c,p)}
:=\left\{
\begin{array}{ll}
+ & (b_c=b_p) \\[2mm]
- & (b_c\neq b_p)
\end{array}
\right.
\end{equation}
is a random variable of an edge parity, where ``$+$'' and ``$-$'' indicate positive and negative parities, respectively, $T=\{e_{j_b}=(c_{j_b},p_{j_b})\}^{|V|-1}_{b=1}$ is an arbitrary spanning tree of given graph $G$, and $\symbffrak{p}_{e_{<b}}:=\{\mathfrak{p}_{e_{1}},\ldots,\mathfrak{p}_{e_{b-1}}\}$.
\red{Note that ${\bf b}^{\rm R}$ can be easily recovered using $T$ and $\mathfrak{p}^{\rm R}(T)$.
}
Here, the reason for considering the probability of the edge parity on the tree is to exclude inconsistent parities such as a triangle with negative parities.
Excluding such an inconsistency does not lose the generality because all possible configurations of solutions can be treated by considering only the probability on the tree.
\red{
Using the edge parity expression, Eq.~\eqref{eq:prop1} can be rewritten as
\begin{equation}
\label{eq:prop1-parity}
\begin{array}{ll}
    \symbffrak{p}^{\rm R}(T)&\displaystyle{=\mathop{\rm argmax}_{\symbffrak{p}\in\{+,-\}^{|V|-1}} \mathcal{P}_m(\symbffrak{p};\rho^*,T)}\\
    &\displaystyle{=\mathop{\rm argmax}_{\symbffrak{p}\in\{+,-\}^{|V|-1}}\prod^{|V|-1}_{b=1} \mathcal{P}_m(\mathfrak{p}_{e_b}|\symbffrak{p}_{e_{<b}};\rho^*,T).}
\end{array}
\end{equation}
}

In the following, we show that QRAO follows Eq.~\red{\eqref{eq:prop1}}, whereas both the RQAOA and RQRAO follow Eq.~\red{\eqref{eq:hypo-general-parity}} to obtain the solution.
All of them assume a moderate approximation on a probability distribution.
The following corollary is frequently used for the proof.

\begin{corollary}
\label{corollary:prob}
Define $\mathcal{E}_{j}:={\rm tr}(P_{\langle j\rangle}\,\rho)$ and $\mathcal{E}_{jk}:={\rm tr}(P_{\langle j\rangle}P_{\langle k\rangle}\,\rho)$.
When the bit string ${\bf b}=b_{|V|-1}\ldots b_1b_0$ is obtained by the $(m,1)$-random magic measurements, ${\bf b}\sim\mathcal{P}_m({\bf b};\rho)$, the probabilities of $b_j=0$ and $b_j=1$ are
\begin{equation}
\label{eq:prob-bj}
\begin{array}{lll}
\mathcal{P}(b_j=0)&=&\displaystyle{\frac{1}{2}+\frac{1}{2\sqrt{m}}\mathcal{E}_{j}}\\[5mm]
\mathcal{P}(b_j=1)&=&\displaystyle{\frac{1}{2}-\frac{1}{2\sqrt{m}}\mathcal{E}_{j}}
\end{array},
\end{equation}
respectively, whereas the probability of $b_j=b_k$ and $b_j\neq b_k$ are
\begin{equation}
\label{eq:prob-bjbk}
\begin{array}{lll}
\mathcal{P}(b_j=b_k)
&=&\left\{
\begin{array}{ll}
\displaystyle{\frac{1}{2}+\frac{1}{2m}\mathcal{E}_{j}\mathcal{E}_{k}} & (q^{(j)}=q^{(k)})\\[5mm]
\displaystyle{\frac{1}{2}+\frac{1}{2m}\mathcal{E}_{jk}} & (q^{(j)}\neq q^{(k)})
\end{array}
\right.\\[10mm]
\mathcal{P}(b_j\neq b_k)
&=&\left\{
\begin{array}{ll}
\displaystyle{\frac{1}{2}-\frac{1}{2m}\mathcal{E}_{j}\mathcal{E}_{k}} & (q^{(j)}=q^{(k)})\\[5mm]
\displaystyle{\frac{1}{2}-\frac{1}{2m}\mathcal{E}_{jk}} & (q^{(j)}\neq q^{(k)})
\end{array}
\right.
\end{array}
,
\end{equation}
respectively, where $q^{(j)}$ is a qubit with $b_j$ embedded.
\end{corollary}
Equation~\eqref{eq:prob-bj} is consistent with the upper bound of the decoding probability, $\frac{1}{2}+\frac{1}{2\sqrt{m}}$~\cite{ambainis1999dense}, which is recovered when $|\mathcal{E}_j|=1$.

\paragraph{Pauli rounding}
\red{
Let us consider two assumptions: bit-flip and independent ones, which we refer to as {\it assumptions A}, denoted by the superscript A.
The first one is the bit-flip assumption
\[\displaystyle{
\mathop{{\rm argmax}}_{{\bf b}} \mathcal{P}^{\rm A}_m({\bf b};\rho^*)
}\]
\begin{equation}
\label{eq:pauli-rounding-ass1}
=\mathop{{\rm argmax}}_{{\bf b}} \mathcal{P}^{\rm A}_m({\bf b};\rho^*) + \mathcal{P}^{\rm A}_m(\overline{\bf b};\rho^*).
\end{equation}
In many cases, this assumption is practically satisfied for \maxcut. This is because the ground states of the \maxcut Hamiltonians are degenerate; ${\bf b}$ and $\overline{\bf b}$ represent the same cut. Moreover, unless a specific constraint is intentionally imposed to balance $\mathcal{P}_m({\bf b};\rho)$ and $\mathcal{P}_m(\overline{\bf b};\rho)$, optimization can result in one of $\mathcal{P}_m({\bf b};\rho)$ or $\mathcal{P}_m(\overline{\bf b};\rho)$ becoming large while the other becomes close to zero.
Alternatively, in the case of an ansatz with $\mathbb{Z}_2$ symmetry~\cite{bravyi2020obstacles}, such as QAOA, the bit-flipped solution appears with the same probability, so the assumption is also satisfied in this case.}

The second one is the independent assumption that all $b_{j_a}$ for $a=1,2,\ldots,|V|$ are independent \red{where $\{j_a\}^{|V|}_{a=1}=\{1,...,|V|\}$, which is equivalent to write the joint distribution as
\begin{equation}
\label{eq:pauli-rounding-ass2}
{\mathcal{P}}^{\rm A}_m({\bf b};\rho^*) = \prod^{|V|}_{a=1}\mathcal{P}^{\rm A}_m(b_{j_a};\rho^*).
\end{equation}
}
Under \red{these} assumption\red{s}, Eq.~\red{\eqref{eq:prop1}} is approximated as
\begin{equation}
\label{eq:pauli-rounding}
\begin{array}{lll}
    b^{\red{{\rm A}}}_{j_a} &=& \displaystyle{\mathop{\rm argmax}_{b_{j_a}\in\{0,1\}} \mathcal{P}^{\red{{\rm A}}}_m(b_{j_a};\rho^*)} \\[5mm]
    &=&\displaystyle{\mathop{\rm argmax}_{b_{j_a}\in\{0,1\}} \Bigg(\frac{1}{2}+\frac{(-1)^{b_{j_a}}}{2\sqrt{m}}\mathcal{E}_{j_a}\Bigg)}\\[5mm]
    &=&\left\{
    \begin{array}{ll}
    0 & (\mathcal{E}_{j_a}>0)\\
    1 & (\mathcal{E}_{j_a}<0)
    \end{array}
    \right.,
\end{array}
\end{equation}
where the second equality uses Corollary~\ref{corollary:prob}.
Interestingly, ${\bf b}^{\red{{\rm A}}}$ is equivalent to the Pauli rounding of QRAO~\cite{fuller2021approximate} \red{defined in Eq.~\eqref{eq:pauli-rounding-org}}.
This comes out that the Pauli rounding seeks the variable with the highest probability of ${\mathcal{P}}_m({\bf b};\rho)$ under the assumption\red{s} of Eq\red{s}.~\red{\eqref{eq:pauli-rounding-ass1} and \eqref{eq:pauli-rounding-ass2}}.

\paragraph{Tree rounding}
Assume all $\mathfrak{p}_{e_b}$ for $b=1,2,\ldots,|V|-1$ are independent\red{, which we refer to as {\it assumption B}, denoted by the superscript B}.
Under this assumption, Eq.~\eqref{eq:hypo-general-parity} is approximated as
\begin{equation}
\label{eq:ass-tree}
{\mathcal{P}}^{\red{{\rm B}}}_{m}(\symbffrak{p};\rho^*,T)
=\displaystyle{\prod_{b=1}^{|V|-1} \mathcal{P}^{\red{{\rm B}}}_m(\mathfrak{p}_{e_b};\rho^*,T)},\\[5mm]
\end{equation}
\red{and Eq.~\red{\eqref{eq:prop1-parity}} is approximated as}
\[\displaystyle{
}\]
\begin{equation}
\label{eq:tree-rounding}
\begin{array}{lll}
    \mathfrak{p}^{\red{{\rm B}}}_{e_b=(c,p)}\red{(T)}
    &=&\displaystyle{\mathop{\rm argmax}_{\mathfrak{p}_{e_b}\in\{+,-\}} \mathcal{P}^{\red{{\rm B}}}_m(\mathfrak{p}_{e_b};\rho^*,\red{T})}\\[5mm]
    &=&\displaystyle{\mathop{\rm argmax}_{(b_c,b_p)\in\{0,1\}^2} 
    \Bigg(\frac{1}{2}+\frac{(-1)^{b_{c}+b_{p}}}{2m}\mathcal{E}_{cp}\Bigg)}\\[5mm]
    &=&\left\{
    \begin{array}{ll}
    + & (\mathcal{E}_{cp}>0)\\
    - & (\mathcal{E}_{cp}<0)
    \end{array}
    \right.,
\end{array}
\end{equation}
where the second equality uses Corollary~\ref{corollary:prob}.
We refer to this rounding method as {\it tree rounding}.

RQRAO with $N=1$, which has no recursive steps, is equivalent to using tree rounding.
\red{
When $N = 1$, the ensemble edge energy $\mathfrak{E}_{jk}$ in Eq.~\eqref{eq:def-energy} is equal to the edge energy estimate $\mathcal{E}_{jk}$ in Eq.~\eqref{eq:def-energy-t} from a single evaluation, and almost surely it takes a non-zero value.
Since the spanning tree is constructed by gathering edges with $\mathfrak{E}_{jk} \neq 0$, a spanning tree covering all nodes can be formed with high probability.
By progressively removing nodes from the leaves to the root of this spanning tree as illustrated in Fig.~\ref{fig:trees}, all nodes except the root are eliminated, determining all parities at once.
This is the same rounding method referred to as the tree rounding defined above.
Since both Pauli rounding and tree rounding assume that all random variables are independent, it is expected that QRAO and RQRAO with $N=1$ will yield similar cut weights.
}

\paragraph{Recursive rounding}
\red{
Let us impose the following two assumptions, which we refer to as \textit{assumptions C}, denoted by the superscript C:
\[\displaystyle{
\max_{\symbffrak{p}\in\{+,-\}^{|V|-1}}\prod^{|V|-1}_{b=1}\mathcal{P}^{\rm C}_m(\mathfrak{p}_{e_b}|\symbffrak{p}_{<e_{b}};\rho^*,T)
}\]
\begin{equation}
=\prod^{|V|-1}_{b=1}\max_{\mathfrak{p}_{e_b}\in\{+,-\}}\mathcal{P}^{\rm C}_m(\mathfrak{p}_{e_b}|\symbffrak{p}_{<e_{b}};\rho^*,T),
\end{equation}
\begin{equation}
\label{eq:assum-c1}
\mathcal{P}^{\rm C}_m(\mathfrak{p}_{e_b}|\symbffrak{p}_{<e_{b}};\rho^*,T)=\mathcal{P}^{\rm C}_m(\mathfrak{p}_{e_b};\rho^*_b(\symbffrak{p}_{<e_{b}}),T),
\end{equation}
where $\rho^*_b(\symbffrak{p}_{<e_{b}})$ is the optimized quantum state that involves the information of $\symbffrak{p}_{<e_{b}}$.
Under these assumptions, Eq.~\eqref{eq:prop1-parity} can be approximated as
\[\displaystyle{
\mathfrak{p}^{\rm C}_{e_b}(T)= \mathop{\rm argmax}_{\mathfrak{p}_{e_b}}\mathcal{P}^{\red{{\rm C}}}_m(\mathfrak{p}_{e_b};\rho^*_b(\symbffrak{p}_{<e_{b}}))
}\]
\begin{equation}
\label{eq:assum-c2}
\ \ \ {\rm for} \ \ \ b=1,\ldots,|V|-1.
\end{equation}
}

\red{
Compared to the Pauli and tree roundings, which impose the strong approximation on $\mathcal{P}_m$ that all random variables are independent, assumptions in Eqs.~\eqref{eq:assum-c1} and \eqref{eq:assum-c2} are weak because each random variable is no longer independent.
Although the assumption of independence in Pauli rounding and tree rounding was made due to the difficulty of conditional sampling, this approach addresses the problem by imposing conditions $\symbffrak{p}_{<e_b}=\symbffrak{p}^*_{<e_b}$ on the quantum state itself instead of on $\mathcal{P}_m$, where $\symbffrak{p}^*_{<e_b}$ is the determined edge parities before step $b$. 
Of course, this approach is not exactly equivalent to the original one, but imposing the condition as a constraint in the optimization should provide a good approximation to conditional sampling.
In the following, we specifically demonstrate how the method of recursively decoding bit strings, as adopted in Recursive QAOA~\cite{bravyi2020obstacles}, follows Eq.~\eqref{eq:prop1-parity} to obtain ${\symbffrak{p}^{\rm R}(T)}$ under the assumptions in Eqs.~\eqref{eq:assum-c1} and \eqref{eq:assum-c2}.
}

First, one parity of $e_1$ is determined using the optimization result, $\rho^*_1$ as
\begin{equation}
\mathfrak{p}^*_{e_1}= \mathop{\rm argmax}_{\mathfrak{p}_{e_1}\in\{+,-\}} \mathcal{P}_m(\mathfrak{p}_{e_1};\rho^*_1),
\end{equation}
where
\begin{equation}
e_1= \mathop{\rm argmax}_{e} \max_{\mathfrak{p}_{e}\in\{+,-\}} \mathcal{P}_m(\mathfrak{p}_{e};\rho^*_1).
\end{equation}
According to Corollary~\ref{corollary:prob}, this can be rewritten as
\begin{equation}
\mathfrak{p}^{\rm rec}_{e_1=(c_1,p_1)}= \left\{
\begin{array}{ll}
+ & (\mathcal{E}_{c_1p_1}>0) \\
- & (\mathcal{E}_{c_1p_1}<0)
\end{array}
\right. ,
\end{equation}
where
\begin{equation}
\begin{array}{lll}
(c_1,p_1)
&=&\displaystyle{\mathop{\rm argmax}_{(c,p)\in E} \Bigg(\frac{1}{2}+\frac{1}{2m}\big|\mathcal{E}_{cp}\big|\Bigg)}\\[5mm]
&=&\displaystyle{\mathop{\rm argmax}_{(c,p)\in E} \big|\mathcal{E}_{cp}\big|}.
\end{array}
\end{equation}
After that, the given graph $G$ is modified to satisfy $\mathfrak{p}_{e_1}=\mathfrak{p}^*_{e_1}$, which refers to $G_2$.
Next, the second optimization is carried out with the $(m,1)$-QRAC Hamiltonian based on $G_2$, and the optimized state $\rho^*_2$ is obtained.
Because the modified graph is forced to be $\mathfrak{p}_{e_1}=\mathfrak{p}^*_{e_1}$, it can be seen that $\rho^*_2$ is conditioned on $\mathfrak{p}_{e_1}=\mathfrak{p}^*_{e_1}$.
Hence, the determined parity of the second step can be written as
\begin{equation}
\label{eq:rec1}
\mathfrak{p}^*_{e_2}= \mathop{\rm argmax}_{\mathfrak{p}_{e_2}\in\{+,-\}}\mathcal{P}_m(\mathfrak{p}_{e_2};\rho^*_2(\mathfrak{p}_{e_1}=\mathfrak{p}^*_{e_1}))
\end{equation}
where
\begin{equation}
\label{eq:rec2}
e_2= \mathop{\rm argmax}_{e} \max_{\mathfrak{p}_{e}\in\{+,-\}} \mathcal{P}_m(\mathfrak{p}_{e};\rho^*_2(\mathfrak{p}_{e_1}=\mathfrak{p}^*_{e_1})).
\end{equation}
Of course, this could also be rewritten using the edge energy $\mathcal{E}_{cp}$, but we will leave that out as it is self-explanatory.
Repeating this procedure gives
\[\displaystyle{
\mathfrak{p}^*_{e_b}= \mathop{\rm argmax}_{\mathfrak{p}_{e_b}}\mathcal{P}_m(\mathfrak{p}_{e_b};\rho^*_b(\symbffrak{p}_{<e_{b}}=\symbffrak{p}^*_{<e_{b}}))
}\]
\begin{equation}
\label{eq:argmax-parity-approx}
\ \ \ {\rm for} \ \ \ b=1,\ldots,|V|-1.
\end{equation}
Finally, $\symbffrak{p}^*=\{\mathfrak{p}^*_1,\ldots,\mathfrak{p}^*_{|V|-1}\}$ is obtained.
The corresponding bit string ${\bf b}^{\rm \red{C}}$ can be easily recovered from $\symbffrak{p}^*$ by fixing one of the $b^{\rm \red{C}}_j$ to $0$ or $1$.
The condition $\symbffrak{p}_{<e_b}=\symbffrak{p}^*_{<e_b}$ is indirectly imposed through $\rho^*_b$ to $\mathcal{P}_m$, and hence this rounding method imposes an approximation that is weaker than that of Pauli and tree rounding.
We refer to this rounding method as {\it recursive rounding}.

In recursive rounding, the edge that determines the parity is selected using the absolute value of its edge energy (Eq.~\eqref{eq:rec2}), and its parity is determined by the sign of the corresponding edge energy (Eq.~\eqref{eq:rec1}).
This process is the same as explained in Sec.~\ref{sec:rqaoa} and hence, it can be said that the RQAOA adopts recursive rounding.
The main idea behind RQAOA is to approximate $\mathop{\rm max}_{{\bf b}\in\{0,1\}^{|V|}}{\rm tr}(H_1|{\bf b}\rangle\langle{\bf b}|)$~\cite{bravyi2020obstacles}, but in actuality, it approximates $\mathop{\rm max}_{{\bf b}\in\{0,1\}^{|V|}}\mathcal{P}_1({\bf b};\rho^*)=\mathop{\rm max}_{{\bf b}\in\{0,1\}^{|V|}}{\rm tr}(\rho^* |{\bf b}\rangle\langle{\bf b}|
)$.\footnote{
Since $\mu_1({\bf b})=|{\bf b}\rangle\langle{\bf b}|$, $f_1({\bf b},\rho)={\rm tr}(\mu_1({\bf b})\rho)$, and $\sum_{{\bf b}\in\{0,1\}^{|V|}}f_1({\bf b},\rho)=1$, $\mathcal{P}_1({\bf b};\rho)=f_1({\bf b},\rho)/\sum_{{\bf b}'}f_1({\bf b}',\rho)={\rm tr}(|{\bf b}\rangle\langle{\bf b}|\rho)$.
}
By contrast, the rounding for RQRAO with $N>1$ can be considered to be an intermediate approach between tree rounding and recursive rounding because the edge parities that are being determined simultaneously in one step are assumed to be independent.

\section{Numerical Experiments}
\label{sec:experiments}

In this section, we show the results of the numerical experiments of the \maxcut problem using RQRAO and other existing algorithms.
The main claim is that RQRAO achieves a better cut weight than existing algorithms and is comparable to the state-of-the-art classical heuristic.

\paragraph{Datasets}
We used three graph datasets.

The first dataset is used to show the validity of Theorem~\ref{theorem:objective}.
That is, it shows that the bit string with the highest probability yields a good cut weight after maximizing the expectation of the $(m,1)$-QRAC Hamiltonian.
Only one graph instance is generated by a machine-independent graph generator \texttt{rudy}~\cite{rudy} using \texttt{-rnd\_graph 14 50 0 -random 0 1 0 -times 2 -plus -1}, where the number of nodes is $14$, the graph density is $0.5$, and the edge weights are set to $\pm1$ uniformly at random.
This instance has only one \maxcut solution, ${\bf b}=00001101101010$, except for its bit-flip solution, where the maximum cut weight is 12.
We call this graph instance \texttt{Rnd14}.

The second dataset is the benchmark graph dataset Gset~\cite{helmberg2000spectral} (\url{http://www.stanford.edu/~yyye/yyye/Gset/}).
Gset includes random, toroidal, and almost planar graphs with $|V|=800$ to $20000$ indexed by G1 to G81, which are generated by \texttt{rudy}.
Only G1, G6, \red{G11,} G14, and G18 are used for evaluation, all of which have 800 nodes.
G1 and G6 are random graphs, G11 is a toroidal graph, and G14 and G18 are the combined graphs of two planar graphs.
Graphs G1 and G6, and graphs G14 and G18 have the same edge structure except for the weights, where G1 and G14 have $+1$ edge weights whereas G6 and G18 have edge weights of $\pm 1$ assigned at random.
The edge weights of G11 are $\pm1$ assigned at random.
We call this graph dataset \texttt{Gset800}.

The third dataset consists of 3-regular graphs with $\pm1$ edge weights generated at random.
The number of nodes varies from $10$ to $10^4$.
Ten graphs are generated for each number of nodes.
These graphs are used to evaluate the scalability with respect to the cut weight and the runtime of the algorithms.
We call this graph dataset \texttt{Rnd3R}.

\paragraph{Baselines}
We compared our model with GW~\cite{goemans1995improved}, 
CirCut~\cite{burer2002rank}, QRAO~\cite{fuller2021approximate}, and RQAOA~\cite{bravyi2020obstacles}.
GW is a well-known classical approximation algorithm based on semidefinite programming.
CirCut is the Fortran program name of the rank-two relaxation algorithm; it is a classical heuristic algorithm that empirically outperforms other classical heuristics on various graph instances~\cite{dunning2018works}.
The details of GW and CirCut are described in \red{Supplementary Information}~\ref{sec:classical-maxcut}.
The RQAOA is a QAOA-based algorithm whose details are explained in Sec.~\ref{sec:rqaoa}.
QRAO is a quantum relaxation-based optimization method whose details are explained in Sec.~\ref{sec:qrao}.

\paragraph{Implementation Details}
The matrix product state (MPS) ansatz is used for $|\psi({\boldsymbol\theta})\rangle$, where all the matrices in the MPS are treated as trainable parameters (see Supplementary Information~\ref{sec:mps} for details).
The default hyperparameters for RQRAO used in the numerical experiments are listed in Table~\ref{tab:hyps}.

We implemented our model in Python, and the MPS optimization is implemented using PyTorch version 1.10.0~\cite{paszke2019pytorch}.
For the optimizer, we used L-BFGS~\cite{liu1989limited} with line search using strong Wolfe conditions.
The termination tolerance for changing the function value and parameters is set to $10^{-2}$.
Implementation details for GW, CirCut, the RQAOA, and QRAO are described in Supplementary Information~\ref{sec:implement-details}.

All experiments were carried out on a single Intel(R) Core(TM) i9-9900X CPU @ 3.50GHz with 32GB (8GBx4) DDR4-2666 Quad-Channel.

\begin{table}[tb]
\centering
\caption{Default hyperparameters for RQRAO.}
\label{tab:hyps}
\begin{tabular}{lcc}
\toprule
Hyperparameter & Symbol & Value \\
\midrule
\# embeddings / qubits & $m$ & $3$ \\
\# ensemble & $N$ & $20$ \\
Scale factor & $S$ & $2$ \\
Bond dimension & $\chi$ & $2$ \\
Brute-force search threshold & $M$ & $10$ \\
\bottomrule
\end{tabular}
\end{table}

\paragraph{Details of the Experiments}

Using the \texttt{Rnd14} dataset, Theorem~\ref{theorem:objective} was verified.
First, $(m,1)$-QRAC Hamiltonians were randomly generated 100 times each for $m=1, 2,$ and $3$.
For each Hamiltonian, the eigenstate with the maximum eigenvalue and optimized state with bond dimension $\chi=2$ were then calculated.
Using each state,
the cut weight with the highest probability was calculated.
We call this experiment \texttt{ExpVer}.

Using the \texttt{Gset800} dataset, the procedure was run 10 times on the same graph instance and the best cut weight of the trials was recorded.
RQRAO and all baseline methods were used.
We call this experiment \texttt{ExpGset}.

Using the \texttt{Rnd3R} dataset, the procedure was run once for each graph instance and the cut weight and elapsed time were recorded.
The cut weight is normalized with respect to that of the GW algorithm and is referred to as the relative cut weight.
Using \texttt{Rnd3R}, two experiments were carried out.
The first one evaluated the scalability.
RQRAO and all baseline methods were used.
We call this experiment \texttt{ExpScale}.
The second one analyzed the sensitivity of the RQRAO hyperparameters.
Only one hyperparameter of RQRAO was varied from its default value (Table~\ref{tab:hyps}).
We call this experiment \texttt{ExpHyp}.

\subsection{Results and Discussion}

\paragraph{Verification of Theorem~\ref{theorem:objective} (\texttt{ExpVer})}

\begin{figure}[t]
    \centering
    \subfigure[MaxEig]{
        \centering
        \includegraphics[scale=1]{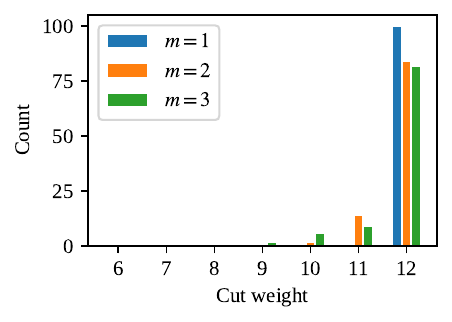}
        \label{fig:largest-maxeig}
    }
    \subfigure[$\chi=2$]{
        \centering
        \includegraphics[scale=1]{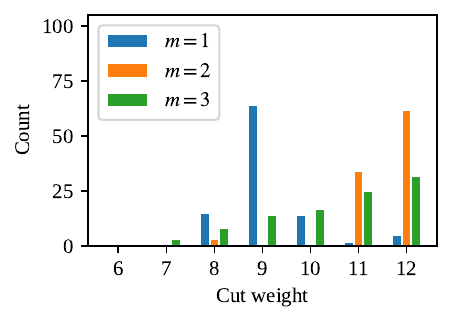}
        \label{fig:largest-chi2}
    }
    \caption{
    Solutions with the highest probability of 100 trials for the same graph instance, \texttt{Rnd14}, where the state used was (a) the eigenstate corresponding to the maximum eigenvalue and (b) the state obtained through optimization with $\chi=2$.
    }
  \label{fig:expver}
\end{figure}

\begin{table*}[tb]
\centering
\caption{Comparison with existing methods on \maxcut problems from the Gset. All instances have $|V|=800$. 
All best-known values were reported in \cite{benlic2013breakout} with the breakout local search algorithm.
GW: Goemans--Williamson algorithm~\cite{goemans1995improved},
RQAOA: recursive quantum approximate optimization algorithm~\cite{bravyi2020obstacles},
QRAO: quantum random access optimization~\cite{fuller2021approximate},
RQRAO: recursive quantum random access optimization (proposed).}
\label{tab:res-gset}
\begin{tabular}{lrrrrrrrrr}
\toprule
Instance & Best-known & GW & CirCut & QRAO & RQAOA & ${\rm RQRAO}_{N=1}$ & ${\rm RQRAO}_{N=20}$\\
\midrule
G1 & 11624 & 11467 & 11622 & 11453 & \red{11460} & 11466 & 11562 \\
G6 &  2178 & 2013 & 2178 & 1980 & \red{1821} & 1980 & 2148 \\
G11 & 564 & 536 & 556 & 540 & \red{560} & 544 & 564 \\
G14 & 3064 & 2999 & 3052 & 2971 & 2965 & 2965 & 3043 \\
G18 & 992 & 924 & 987 & 929 & 882 & 920 & 980 \\
\bottomrule
\end{tabular}
\end{table*}
\begin{figure*}[t]
    \centering
    \subfigure[]{
        \centering
        \includegraphics[scale=.8]{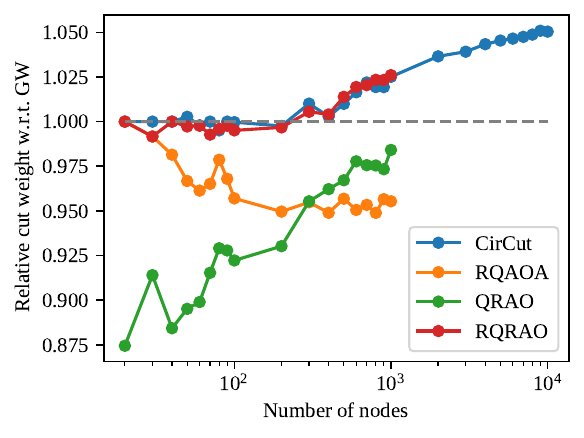}
        \label{fig:scalability-cut}
    }
    \subfigure[]{
        \centering
        \includegraphics[scale=.8]{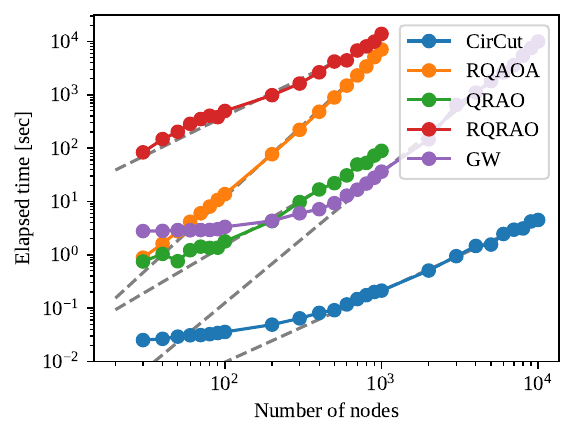}
        \label{fig:scalability-time}
    }
    \caption{Comparison of scalability with existing methods. Ten 3-regular graphs with the same number of nodes were randomly generated whose edge weights were $\pm1$ at random. 
    The markers indicate the means.
    (a) Relative cut weight with respect to the cut weight obtained by the GW algorithm; (b) Algorithm run time. The gray dashed line is a fitted curve using $({\rm Elapsed\ time\ [sec]})=a({\rm Number\ of\ nodes})^b$}.
  \label{fig:scalability}
\end{figure*}

\begin{figure*}[t]
    \subfigure[]{
        \centering
        \includegraphics[scale=.8]{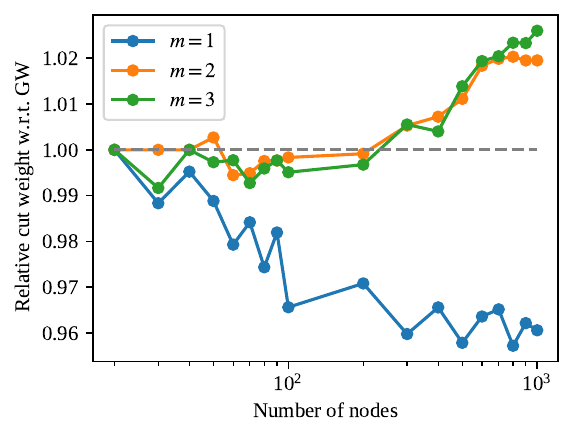}
        \label{fig:scalability-cut-m}
    }
    \subfigure[]{
        \centering
        \includegraphics[scale=.8]{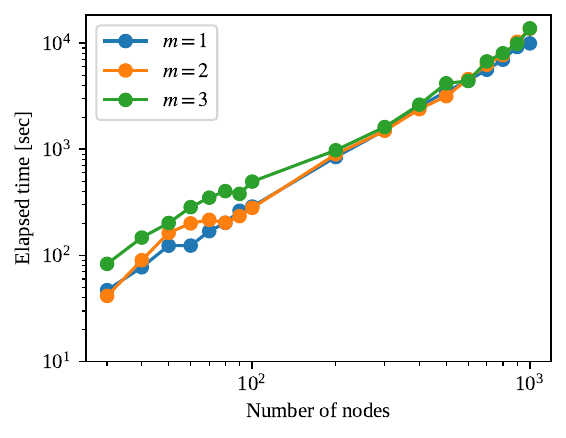}
        \label{fig:scalability-time-m}
    }
    \caption{Results of the ablation study for the number of embeddings per qubit $m$. The problem graph instances are $3$-regular graphs with $\pm1$ edge weights. For each number of nodes, $10$ graph instances were randomly generated.
    }
  \label{fig:ablation}
\end{figure*}

\red{
Figure~\ref{fig:expver} summarizes the results of \texttt{ExpVer}.
When the eigenstate with the maximum eigenvalue is obtained, the $(1,1)$-QRAC is consistently the best choice as it provides the optimal \maxcut solution with the highest probability (Fig.~\ref{fig:largest-maxeig}).
On the other hand, for $(2,1)$ and $(3,1)$-QRACs, particularly when higher cut weight solutions cluster in a different location within the Hilbert space from the \maxcut solution, the solution with the highest probability may not correspond to the optimal \maxcut solution.
Therefore, even if the ground state is obtained, there is a disadvantage that the nearest solution may not be the optimal \maxcut solution.}

\red{
Nevertheless, when the quantum state is obtained through optimization (Fig.~\ref{fig:largest-chi2}), $(2,1)$ and $(3,1)$-QRACs achieve higher cut weights than when using $(1,1)$-QRAC.
The reason why higher cut weights can be obtained by setting $m\ge2$ 
can be explained as follows.
As shown in Fig.~\ref{fig:rqrao-concept}, as $m$ increases, more solutions are packed into \red{fewer qubits}, resulting in an average over a greater number of solutions being observed.
In other words, as $m$ increases, we can only observe the average of a larger number of solutions, leading to a \textit{blurred} expectation value.
Since the optimization process attempts to reach the optimal state from such a blurred function value that only reflects the average, increasing $m$ does not necessarily bring the optimized state closer to the nearest state embedding the optimal solution.
However, this blurring could also potentially make the function value smoother, thus simplifying the optimization process.
Such optimization-related advantages occur for $m=2$ and $m=3$, and it is likely that $m=2$ resulted in the highest cut weight because it strikes a balance between the aforementioned advantages and disadvantages.
}

\red{
In this study, we have optimized a classically trackable MPS ansatz using classical optimization methods.
However, recent developments, such as quantum state preparation using Lindbladians~\cite{ding2024single}, suggest that \textit{pure} quantum algorithms are promising.
In our method, more than 90\% of the computation time is spent on parameter optimization, so there could be benefits in terms of overall algorithm speedup by replacing the optimization process with a pure quantum algorithm.
However, since the discussion on optimization is beyond the scope of this study, these considerations will be addressed in future work.
In any case, the \maxcut problem, traditionally a problem of finding the eigenstate with the maximum eigenvalue, has effectively been shifted to finding the magic state embedding with the highest fidelity to the optimized quantum state.
While challenging, this approach may yield better solutions, particularly with $(2,1)$ and $(3,1)$-QRACs.
In the following, we will explore the advantages of $m \ge 2$ through additional experiments.
}

\paragraph{Gset (\texttt{ExpGset})}
The cut weight on the Gset dataset for each algorithm is summarized in Table~\ref{tab:res-gset}.
RQRAO with $N=20$ exceeds both the classical approximation algorithm (GW) and quantum algorithms (QRAO and the RQAOA), but is slightly worse than the classical heuristic algorithm (CirCut) on all graphs except for G11.

As expected in Sec.~\ref{sec:prop1}, the cut weight of RQRAO with $N=1$ is very similar to that of QRAO because both assume all random variables are independent.
The cut weight performance of RQRAO increases when the number of ensembles $N$ is increased, which confirms the effectiveness of using an ensemble.

\paragraph{Scalability (\texttt{ExpScale})}

Figure~\ref{fig:scalability} summarizes the results on $3$-regular graphs.
Although the relative \red{cut weights} of QRAO and the RQAOA are lower than 1, that of RQRAO exceeds 1 when the number of nodes is larger than 200 and is larger than that of CirCut (Fig.~\ref{fig:scalability-cut}).
In addition, although the absolute runtime of RQRAO is longer than that of CirCut, the order of empirical time complexity of these methods is close (Fig.~\ref{fig:scalability-time}).
To quantitatively compare the empirical time complexity of the methods, the runtime curves were fitted using 
\[\displaystyle{
({\rm Wall\ time\ [sec]})=\alpha ({\rm Number\ of\ nodes})^{\beta}
}\]
using only data with a large number of nodes.
The fitted values of $\beta$ are $1.34$ for CirCut, $\red{2.71}$ for the RQAOA, $1.73$ for QRAO, $1.44$ for RQRAO, and $2.44$ for GW.
That is, the empirical time complexity of the proposed method is much lower than that of GW and the RQAOA, and comparable to that of CirCut.
Because most of the computational time is spent on optimizing the variational states, the computational time of the proposed method can be reduced by improving the optimization algorithm.
For example, the runtime cost of RQRAO can be reduced by parallel computation in the ensemble part.
Further reductions in the computational costs of RQRAO are future work.
\red{
In Supplementary Information~\ref{sec:2d-toric}, we also compared the methods using the 2D toric graphs with one extra node connected to all other nodes, as used in~\cite{bravyi2020obstacles} to demonstrate the power of RQAOA.
The results reaffirm what we have observed when comparing the methods using 3-regular graphs.
}

\paragraph{Sensitivity Analysis (\texttt{ExpHyp})}

Only the results for varying the number of embeddings per qubit $m$ are described here.
Other hyperparameter sensitivities \red{such as the number of ensemble $N$, the scale factor $S$, and the bond dimension $\chi$,} are summarized in Supplemental Information~\ref{sec:ablation-others}.
In Fig.~\ref{fig:scalability-cut-m}, we can see the clear advantage of using $m>1$ over $m=1$; only the cut weight in the cases of $m>1$ overcomes that of GW.
When $m=1$, the optimized state may be stacked in \red{bad} local optima, as shown in Fig.~\ref{fig:largest-chi2}, and this is the reason for the lower cut weight.
Almost no runtime dependency on $m$ is seen in Fig.~\ref{fig:scalability-time-m}, which means there is no disadvantage to use $m>1$ as long as the given graph is sufficiently sparse to allow applying $(m,1)$-QRAC Hamiltonian.

\section{Conclusion}
\label{sec:conclusion}
In this study, we proposed a recursive and quantum-relaxed algorithm to solve the \maxcut problem.
The quantum relaxation is achieved using the QRAC, which is incorporated into the existing recursive \maxcut algorithm.
To show the validity of the proposed algorithm, new properties of the quantum-relaxed Hamiltonian were revealed.
According to these findings, we showed that the \maxcut problem can be approximately reduced to the problem finding the variable with the highest probability, which is determined by the quantum state optimized through the quantum-relaxed Hamiltonian.
We confirmed that the proposed algorithm is consistent with this objective and includes the existing methods as special cases.
The numerical experiments on graphs with several hundred nodes showed that the proposed algorithm outperforms the GW algorithm, which is the best-known classical approximation algorithm, on graphs with more than 300 nodes both in the approximation ratio and the empirical time complexity.

\section*{Acknowledgement}
\label{sec:acknowledgement}
This work is supported by MEXT Quantum Leap Flagship Program Grant Number JPMXS0118067285 and JPMXS0120319794, and JSPS KAKENHI Grant Number 20H05966. 
We would like to thank Yohichi Suzuki and Michihiko Sugawara of Keio Quantum Computing Center and Hiroshi C. Watanabe of Kyushu University for discussing the theoretical and experimental justification of RQRAO.
RR would like to acknowledge Prof. Hiroshi Imai for information on Gset instances and \maxcut. 

\bibliographystyle{unsrtnat} 
\bibliography{09_reference}

\onecolumn\newpage

\appendix
\renewcommand{\thesubsection}{\Alph{section}.\arabic{subsection}}

\setcounter{equation}{0}
\renewcommand{\theequation}{S.\arabic{equation}}
\setcounter{figure}{0}
\renewcommand{\thefigure}{S.\arabic{figure}}

\makeatletter
\renewcommand\p@subfigure{\thefigure}
\makeatother

\section{\red{Random Magic Measurements and the Relationship with Classical Shadows}}
\label{sec:shadow-magic}

Let us define $(m,1)$-random magic measurements, which is referred to as {\it magic state rounding} in \cite{fuller2021approximate}.
In this section, only $(3,1)$-random magic measurements are introduced, as $(2,1)$- and $(1,1)$-random magic measurements are easily derived similarly.

\begin{definition}[Random magic measurements~\cite{fuller2021approximate}]
\label{def:magic-meas}
    Let $\mu^{[q]\pm}_i$ be the single qubit $(3,1)$-magic state corresponding to the $q$th qubit defined as follows:
    \[\displaystyle{
    \mu_1^{[q]+} := \mu_3^{[q]}(0,0,0),\ \ \ \mu_1^{[q]-} := \mu_3^{[q]}(1,1,1),
    }\]
    \[\displaystyle{
    \mu_2^{[q]+} := \mu_3^{[q]}(0,1,1),\ \ \ \mu_2^{[q]-} := \mu_3^{[q]}(1,0,0),
    }\]
    \[\displaystyle{
    \mu_3^{[q]+} := \mu_3^{[q]}(1,0,1),\ \ \ \mu_3^{[q]-} := \mu_3^{[q]}(0,1,0),
    }\]
    \[\displaystyle{
    \mu_4^{[q]+} := \mu_3^{[q]}(1,1,0),\ \ \ \mu_4^{[q]-} := \mu_3^{[q]}(0,0,1),
    }\]
    where $\mu_3^{[q]}(b^{[q]}_{X},b^{[q]}_{Y},b^{[q]}_{Z})$ is defined as Eq.~\eqref{eq:def-magic-state}.
    Random magic measurement is the protocol that obtains ${\bf b}\in\{0,1\}^{3n}$ as follows:
    \begin{enumerate}
    \item Select $i\in\{1,2,3,4\}$ uniformly at random for each $q$.
    \item Measure the state $\rho$ in the basis $\{\mu^{[q]\pm}_i\}$ for all qubits $q\in[n]$ with the selected $i$.
    \item Convert the measurement outcomes to the corresponding classical bits as listed in Table~\ref{tab:outcome-bits}.
    \end{enumerate}
\end{definition}

\begin{table}[tb]
\centering
\caption{Measurement outcome of the random magic measurement and corresponding classical bits. Here, $i_q$ is the selected $i\in\{1,2,3,4\}$ for qubit $q$.}
\label{tab:outcome-bits}
\begin{tabular}{ccc}
\toprule
$i_q$ & Outcome & $\big(b^{[q]}_X,b^{[q]}_Y,b^{[q]}_Z\big)$\\
\midrule
1 & 0 & $(0,0,0)$ \\
1 & 1 & $(1,1,1)$ \\
2 & 0 & $(0,1,1)$ \\
2 & 1 & $(1,0,0)$ \\
3 & 0 & $(1,0,1)$ \\
3 & 1 & $(0,1,0)$ \\
4 & 0 & $(1,1,0)$ \\
4 & 1 & $(0,0,1)$ \\
\bottomrule
\end{tabular}
\end{table}

The bit-strings ${\bf b}$ obtained by a random magic measurement follows the probability distribution $\mathcal{P}_{\rm f}({\bf b})$ (see the proof of Theorem~\ref{theorem:fidelity}).

Because the expectation of the $k$-local Pauli observable in Eq.~\eqref{eq:k-local-exp} is written as the expectation of the random variable
\begin{equation}
\hat{o} := 3^{k/2} (-1)^{\sum^k_{i=1}b^{[q^{(i)}]}_{P}},
\end{equation}
the classical shadow~\cite{huang2020predicting} based on random magic measurements can be considered.
That is, the expectation value of the $k$-local Pauli observables can be estimated by averaging $\hat{o}$ computed from the corrected bits $\{b^{[q]}_{P}\}$ by the random magic measurements.
The upper bound of the variance of random variable $\hat{o}$ is readily derived as 
\begin{equation}
{\rm Var}[\hat{o}]
=\mathbb{E}[\hat{o}^2]-(\mathbb{E}[\hat{o}])^2
\le\mathbb{E}[\hat{o}^2]
=3^k,
\end{equation}
which is equivalent to that of random Pauli measurements~\cite{huang2020predicting}.
The classical snapshot based on random magic measurement is
\begin{equation}
\label{eq:snapshot}
\hat{\rho}=\bigotimes^{n-1}_{q=0} \Big( 3U^{[q]\dagger}|\hat{b}^{[q]}\rangle \langle\hat{b}^{[q]}|U^{[q]} - I\Big),
\ \ \ U^{[q]}\sim\mathcal{U}_{\rm ms},
\end{equation}
where $\hat{b}^{[q]}$ is the measurement outcome of the $q$th qubit.
Equation~\eqref{eq:snapshot} is the same as that of random Pauli measurements except for the unitary ensemble $\mathcal{U}_{\rm ms}$.
In the case of random magic measurements, the unitary ensemble is $\mathcal{U}_{\rm ms}=\{U_i\}^4_{i=1}$, where
$\mu_i^+:=U_i|0\rangle\langle0|U_i^\dagger$ and $\mu_i^-:=U_i|1\rangle\langle1|U_i^\dagger$.
Note that $U_1$ can be expressed as $U_1^\dagger=e^{\red{-}{\rm i}tX}e^{\red{-}{\rm i}sZ}$ where $s=\pi/8$ and $t={\rm arccos}(\sqrt{(1+1/\sqrt{3})/2})$, and $U_2^\dagger=U^\dagger_1X$, $U^\dagger_3=U^\dagger_1Y$, $U^\dagger_4=U^\dagger_1Z$~\cite{fuller2021approximate}.

\begin{figure*}[t]
    \centering
    \subfigure[]{
        \centering
        \includegraphics[scale=.8]{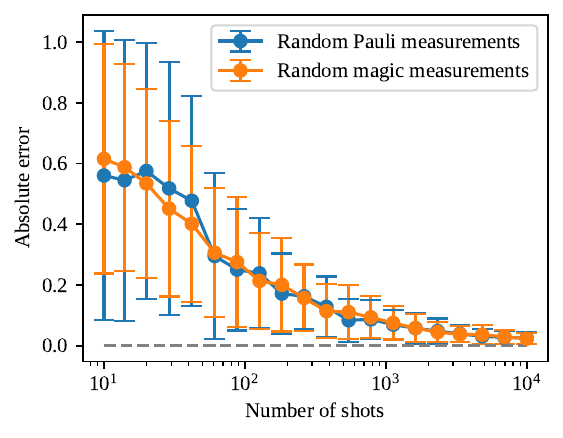}
        \label{fig:shadow-1}
    }
    \subfigure[]{
        \centering
        \includegraphics[scale=.8]{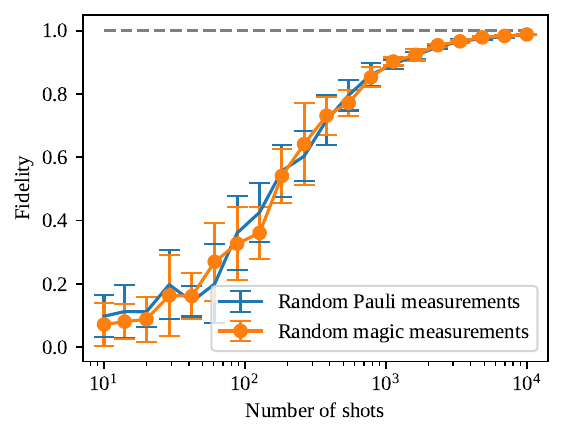}
        \label{fig:shadow-2}
    }
    \caption{Accuracy comparison of a classical shadow obtained using random Pauli measurements and that obtained using $(3,1)$-random magic measurements. (a) Absolute error of the expectation value of a $2$-local Pauli observable in a ten qubit system; (b) Fidelity between the true state and the one obtained using state tomography in a five qubit system.}
  \label{fig:shadow}
\end{figure*}

To demonstrate the performance of random magic measurements, two numerical experiments were carried out.

\paragraph{Expectation Estimation}
A state $\rho$ is sampled from ten-qubit Haar's random states and is used throughout the experiments.
A hundred 2-local Pauli observables are randomly generated and their expectations are evaluated using random Pauli measurements and random magic measurements, respectively.
Although the expectation of the $k$-local Pauli observable is within $[-1,1]$, the range of the expectation estimated by random magic measurements is $[-3^{k/2},3^{k/2}]$.
To mitigate the anomaly estimation, the expectation value estimated by the random magic measurements is truncated so that the value is within $[-1,1]$.
In Fig.~\ref{fig:shadow-1}, the accuracy of the expectation evaluation of 2-local Pauli observable is shown.
The dependence of the accuracy of the expectation evaluation on the number of shots is almost the same for both the random Pauli measurements and random magic measurements, which is as expected because the variance upper bounds of these methods are the same.

\paragraph{State Tomography}
A state $\rho$ is sampled from five-qubit Haar's random states and used throughout the experiments.
$N$ snapshots, $\{\hat{\rho}_i\}^N_{i=1}$, are collected using random Pauli measurements and the random magic measurements.
The candidate matrix is then defined as $\hat{\rho}=\sum^N_{i=1}\hat{\rho}_i/N$.
Here, in general, $\hat{\rho}$ does not satisfy the property of the density matrix, i.e., ${\rm tr}(\hat{\rho})=1$ and $\hat{\rho}\succeq0$.
To avoid improper estimation of the fidelity, the nearest state to $\hat{\rho}$ that satisfied the property of the density matrix is obtained by the optimization, and is defined as the result of the state tomography.
The optimization process is as follows.
The $n$-qubit nearest state is parametrized as $|\hat{\psi}({\boldsymbol\theta})\rangle_k=f_k / \sqrt{\sum^{2^n}_{k'=1} f^*_{k'} f_{k'}}$ where $f_k=\theta^{\rm Re}_k+{\rm i}\theta^{\rm Im}_k$ and ${\boldsymbol\theta}=\{\theta^{\rm Re}_k,\theta^{\rm Im}_k\}^{2^n}_{k=1}$.
Here, $|\hat{\psi}({\boldsymbol\theta})\rangle_k$ denotes the $k$th component of $|\hat{\psi}({\boldsymbol\theta})\rangle$ in the computational basis.
The objective function is defined as $\mathcal{L}({\boldsymbol\theta}):=\|\hat{\rho}-|\hat{\psi}({\boldsymbol\theta})\rangle\langle\hat{\psi}({\boldsymbol\theta})|\|_{\rm F}$ where $\|\bullet\|_{\rm F}$ is the Frobenius norm.
Parameter ${\boldsymbol\theta}$ is optimized by L-BFGS to minimize $\mathcal{L}({\boldsymbol\theta})$.
Then, the tomography result is defined as $\overline{\rho}:=|\hat{\psi}({\boldsymbol\theta}^*)\rangle\langle\hat{\psi}({\boldsymbol\theta}^*)|$, where ${\boldsymbol\theta}^*$ is the optimized parameter.
According to Fig.~\ref{fig:shadow-2}, the dependence of the state tomography accuracy on the number of shots is almost the same for both the random Pauli measurements and random magic measurements.

\section{Node Isolation Problem}
\label{sec:nodes-isolation}
\begin{figure*}[t]
    \centering
    \subfigure[]{
        \centering
        \includegraphics[scale=.8]{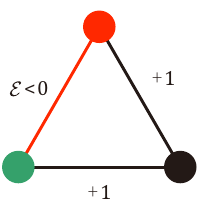}
        \label{fig:edge-noise-a}
    }
    \subfigure[]{
        \centering
        \includegraphics[scale=.8]{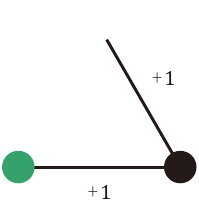}
        \label{fig:edge-noise-b}
    }
    \subfigure[]{
        \centering
        \includegraphics[scale=.8]{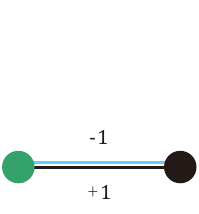}
        \label{fig:edge-noise-c}
    }
    \subfigure[]{
        \centering
        \includegraphics[scale=.8]{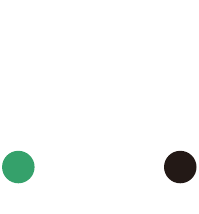}
        \label{fig:edge-noise-d}
    }
    \caption{
    Schematic illustration of the node isolation process.
    (a) The red and green nodes represent the node to be deleted and the node to be kept, respectively. The red edge indicates the candidate edge. The ensemble edge energy of the candidate edge $\mathcal{E}$ is negative.
    (b) After removing the candidate edge.
    (c) Multiplying the sign of the ensemble edge energy of the candidate edge and reconnecting it to the green node.
    (d) After removing the zero-weight edge.
    }
  \label{fig:edge-noise}
\end{figure*}

When the edge weights of problem graph $G$ are integers, the number of isolated nodes, which are defined as nodes with undetermined parity with respect to other nodes, will increase as the iteration progresses.
An example is illustrated in Fig.~\ref{fig:edge-noise}.
In Fig.~\ref{fig:edge-noise-a}, the node to be deleted and the node to be kept are shown in red and green, respectively.
The red edge indicates the candidate whose ensemble edge energy $\mathcal{E}$ is negative.
After removing the candidate edge (Fig.~\ref{fig:edge-noise-b}), the edge that was connected to the red node at one end is reconnected to the green node (Fig.~\ref{fig:edge-noise-c}).
At this time, the edge weight is multiplied by the sign of $\mathcal{E}$.
As a result, the weight of the edge the bottom of the illustration is updated to zero (Fig.~\ref{fig:edge-noise-d}), which means that the black node becomes isolated.
If there are $N_{\rm iso}$ isolated nodes after the parity determination procedure, a brute-force search must be carried out on the $M+N_{\rm iso}$ nodes, which is intractable when there are several tens of isolated nodes.
One easy way to mitigate the node isolation problem is to apply a small perturbation to all edge weights at random.
When perturbed edge weights are used, the probability that a reconnected edge has zero weight becomes negligibly small.

\section{Proofs}
\label{sec:proof}

\begin{proof}[Proof of Theorem~\ref{theorem:fidelity}]
Consider a quantum channel $\mathcal{E}$ that measures an $n$ qubit state $\rho$ with the measurement basis $\ms_m({\bf b})$, which is an $n$ qubit product state of magic states selected uniformly at random:
\begin{equation}
\label{eq:ms-channel}
    \mathcal{E}(\rho)
    :=\frac{1}{2^{(m-1)n}}\sum_{{\bf b}\in\{0,1\}^{mn}}{\rm tr}(\ms_m({\bf b}) \rho)\ms_m({\bf b}),
\end{equation}
where $2^{(m-1)n}$ is the total number of $(m,1)$-magic state bases living in $n$-qubit quantum state space.
According to \cite{fuller2021approximate}, the following holds:
\begin{equation}
\label{eq:ms-rounding}
    {\rm tr}\big(P^{\otimes k}\mathcal{E}(\rho)\big)
    =\frac{1}{m^k}{\rm tr}\big(P^{\otimes k}\rho\big)
\end{equation}
In addition, it is easily confirmed that
\begin{equation}
\label{eq:tr-Pms}
{\rm tr}(P^{\otimes k}\ms_m({\bf b}))=\frac{(-1)^{\sum^k_{i=1} b^{[q^{(i)}]}_{P}}}{m^{k/2}}.
\end{equation}
Substituting Eqs.~\eqref{eq:ms-channel} and \eqref{eq:tr-Pms} into Eq.~\eqref{eq:ms-rounding} yields
\[\displaystyle{
{\rm tr}(P^{\otimes k}\rho)
=m^k{\rm tr}(P^{\otimes k}\mathcal{E}(\rho))
}\]
\[\displaystyle{
=m^k{\rm tr}\Bigg(P^{\otimes k}\frac{1}{2^{(m-1)n}}\sum_{{\bf b}\in\{0,1\}^{mn}}{\rm tr}(\ms_m({\bf b}) \rho)\ms_m({\bf b})\Bigg)
}\]
\[\displaystyle{
=\frac{m^k}{2^{(m-1)n}}\sum_{{\bf b}\in\{0,1\}^{mn}}{\rm tr}(P^{\otimes k}\ms_m({\bf b}) )f_m({\bf b},\rho)
}\]
\begin{equation}
\label{eq:proof-identity}
=\frac{m^{k/2}}{2^{(m-1)n}}\sum_{{\bf b}\in\{0,1\}^{mn}}(-1)^{\sum^k_{i=1} b^{[q^{(i)}]}_{P}}f_m({\bf b},\rho),
\end{equation}
where
\begin{equation}
f_m({\bf b},\rho):={\rm tr}(\ms_m({\bf b})\rho)
\end{equation}
is a fidelity of $\ms_m({\bf b})$ and $\rho$.
Meanwhile, it is easily confirmed that
\begin{equation}
1={\rm tr}(\mathcal{E}(\rho))=\frac{1}{2^{(m-1)n}}\sum_{{\bf b}'\in\{0,1\}^{mn}}f({\bf b}',\rho).
\end{equation}
Then, the following holds:
\[\displaystyle{
{\rm tr}(P^{\otimes k}\rho)
=m^{k/2} \, \frac{\sum_{{\bf b}\in\{0,1\}^{mn}}(-1)^{\sum^k_{i=1} b^{[q^{(i)}]}_P}f_m({\bf b},\rho)}{\sum_{{\bf b}'\in\{0,1\}^{mn}}f_m({\bf b}',\rho)}
}\]
\begin{equation}
=m^{k/2} \, \mathbb{E}_{{\bf b}\sim\mathcal{P}m({\bf b};\rho)} \Big[ (-1)^{\sum^k_{i=1} b^{[q^{(i)}]}_P} \Big],
\end{equation}
where
\begin{equation}
\mathcal{P}_m({\bf b};\rho) := \frac{f_m({\bf b},\rho)}{\sum_{{\bf b}'\in\{0,1\}^{mn}}f_m({\bf b}',\rho)}
\end{equation}
\qed
\end{proof}

\begin{proof}[Proof of Corollary~\ref{corollary:qrac-hamiltonian}]
Consider that a binary variable of node $j$, $b_j$, is mapped to $(q^{(j)}, P^{(j)})$, where $q^{(j)}$ is a qubit index, $P^{(j)}\in\mathfrak{P}\subset\{X,Y,Z\}$ and $|\mathfrak{P}|=m\in\{1,2,3\}$ (see Sec.~\ref{sec:qrao}).
Equation~\eqref{eq:proof-identity} in the case of $k=2$ reduces
\begin{equation}
\label{eq:two-local-identity}
{\rm tr}\big(P^{(j)}_{q^{(j)}}P^{(k)}_{q^{(k)}}\rho\big)
=m\,\mathbb{E}_{{\bf b}\sim\mathcal{P}_m({\bf b};\rho)} \big[ (-1)^{b_j+b_k} \big].
\end{equation}
The left-hand side of Eq.~\eqref{eq:two-local-identity} is written as ${\rm tr}\big(P_{\langle j \rangle}P_{\langle k \rangle}\rho\big)$ in the main body.
Substituting the $(m,1)$-QRAC Hamiltonian (Eq.~\eqref{eq:h-qrac}) into Eq.~\eqref{eq:two-local-identity} leads to
\begin{equation}
\label{eq:Hqrac-w}
    {\rm tr}(H_m \, \rho)=\sum_{(j,k)\in E}\frac{{\rm tr}(\rho)-m{\rm tr}(P_{\langle j \rangle}P_{\langle k \rangle}\,\rho)}{2}w_{jk}
    =\sum_{(j,k)\in E}\frac{1-m^2\mathbb{E}_{{\bf b}\sim\mathcal{P}_m({\bf b};\rho)}\big[(-1)^{b_j+b_k}\big]}{2}w_{jk}.
\end{equation}

According to the definition of $\texttt{CW}({\bf b})$ (Eq.~\eqref{eq:def-cut}),
\begin{equation}
\label{eq:parity-cw}
\sum_{(j,k)\in E}\frac{(-1)^{b_j+b_k}}{2}w_{jk}=\sum_{(j,k)\in E}\frac{1}{2}w_{jk}-\texttt{CW}({\bf b}).
\end{equation}
Substituting Eq.~\eqref{eq:parity-cw} into \eqref{eq:Hqrac-w} gives
\begin{equation}
\label{eq:}
    {\rm tr}(H_{\rm QRAC} \, \rho)
    =m^2\mathbb{E}_{{\bf b}\sim\mathcal{P}_m({\bf b};\rho)}\big[\texttt{CW}({\bf b})]-\frac{m^2-1}{2}\sum_{(j,k)\in E}w_{jk}.
\end{equation}
\qed
\end{proof}

\begin{proof}[Proof of Corollary~\ref{corollary:qrac-variance}]
According to Eq.~\eqref{eq:h-qrac},
\[
\begin{array}{ll}
\displaystyle{H^2_m}
&\displaystyle{= \Big(\sum_{(j,k)\in E}\frac{I-mP_{\langle j \rangle}P_{\langle k \rangle}}{2}w_{jk}\Big)\cdot \Big(\sum_{(j',k')\in E}\frac{I-mP_{\langle j' \rangle}P_{\langle k' \rangle}}{2}w_{j'k'}\Big)} \\[5mm]
&\displaystyle{=\frac{m^2}{4}\sum_{(j,k)\in E}\sum_{(j',k')\in E}P_{\langle j \rangle}P_{\langle k \rangle}P_{\langle j' \rangle}P_{\langle k' \rangle}w_{jk}w_{j'k'}-\frac{mW}{2}\sum_{(j,k)\in E}P_{\langle j \rangle}P_{\langle k \rangle}w_{jk}}+\frac{1}{4}W^2
\end{array},
\]
where $W:=\sum_{(j,k)\in E}w_{jk}$.
The expectation value of $H^2_m$ for an arbitrary state $\rho$ is calculated as
\[
\begin{array}{ll}
\displaystyle{{\rm tr}(H^2_m\rho)}
&\displaystyle{=\frac{m^2}{4}\sum_{(j,k)\in E}\sum_{(j',k')\in E}{\rm tr}(P_{\langle j \rangle}P_{\langle k \rangle}P_{\langle j' \rangle}P_{\langle k' \rangle}w_{jk}w_{j'k'}\rho)-\frac{mW}{2}\sum_{(j,k)\in E}{\rm tr}(P_{\langle j \rangle}P_{\langle k \rangle}w_{jk}}\rho)+\frac{1}{4}W^2\\[5mm]
&\displaystyle{=\frac{m^2}{4}\sum_{(j,k)\in E}\sum_{(j',k')\in E}{\rm tr}(P_{\langle j \rangle}P_{\langle k \rangle}P_{\langle j' \rangle}P_{\langle k' \rangle}w_{jk}w_{j'k'}\rho)+W{\rm tr}(H_m\rho)-\frac{1}{4}W^2}.
\end{array}
\]
Meanwhile, from Eq.~\eqref{eq:tr-qrac},
\[
\begin{array}{ll}
\displaystyle{(\texttt{CW}({\bf b}))^2}
&\displaystyle{= \Big(\sum_{(j,k)\in E}\frac{1-m^2(-1)^{b_j+b_k}}{2m^2}w_{jk}+\frac{m^2-1}{2m^2}W\Big)\Big(\sum_{(j',k')\in E}\frac{1-m^2(-1)^{b_{j'}+b_{k'}}}{2m^2}w_{j'k'}+\frac{m^2-1}{2m^2}W\Big)}\\[5mm]
&\displaystyle{= \frac{1}{4}\sum_{(j,k)\in E}\sum_{(j',k')\in E}(-1)^{b_j+b_k+b_{j'}+b_{k'}}w_{jk}w_{j'k'}-\frac{1}{2}W\sum_{(j,k)\in E}(-1)^{b_j+b_k}w_{jk}+\frac{1}{4}W^2}.
\end{array}
\]
Using Theorem~\ref{theorem:fidelity} and substituting the expression of ${\rm tr}(H^2_m\rho)$ derived above yields
\[
\begin{array}{ll}
\displaystyle{\mathbb{E}_{{\bf b}\sim \mathcal{P}_m({\bf b};\rho)}\big[(\texttt{CW}({\bf b}))^2\big]}
&\displaystyle{=\frac{1}{m^4}{\rm tr}(H^2_m\rho)+\frac{m^2-1}{m^4}W{\rm tr}(W\rho)+\frac{(m^2-1)^2}{4m^2}W^2}.
\end{array}
\]
Combining with Eq.~\eqref{eq:tr-qrac}, the variance of the cut weight can be written as
\[
\begin{array}{ll}
\displaystyle{{\rm Var}_{{\bf b}\sim \mathcal{P}_m({\bf b};\rho)}\big[\texttt{CW}({\bf b})\big]}
&\displaystyle{= \mathbb{E}_{{\bf b}\sim \mathcal{P}_m({\bf b};\rho)}\big[\texttt{CW}({\bf b})^2\big] - \mathbb{E}_{{\bf b}\sim \mathcal{P}_m({\bf b};\rho)}\big[\texttt{CW}({\bf b})\big]^2} \\[5mm]
&\displaystyle{ = \frac{1}{m^4}{\rm tr}(H_m^2\rho)+\frac{m^2-1}{m^4}W{\rm tr}(H_m\rho)+\frac{(m^2-1)^2}{4m^4}W^2} \\[2mm]
&\displaystyle{{\hspace{10mm}}-\Big(\frac{1}{m^2}{\rm tr}(H_m\rho)+\frac{m^2-1}{2m^2}W\Big)^2} \\[5mm]
&\displaystyle{=\frac{1}{m^4}\Big({\rm tr}(H_m^2\rho)-{\rm tr}(H_m\rho)^2\Big)}.
\end{array}
\]
\qed
\end{proof}

\begin{proof}[Proof of Corollary~\ref{corollary:prob}]
According to Eq.~\eqref{eq:k-local-exp},
\[
\begin{array}{ll}
\displaystyle{\mathcal{E}_{j}}
&\displaystyle{:= {\rm tr}(P_{\langle j\rangle}\rho)} \\[2mm]
&\displaystyle{ = \sqrt{m}\,\mathbb{E}_{{\bf b}\sim\mathcal{P}_{\rm f}({\bf b};\rho)}[(-1)^{b_j}]} \\[2mm]
&\displaystyle{=\sqrt{m}\,\big[\mathcal{P}(b_j=0)\cdot(+1)+\mathcal{P}(b_j=1)\cdot(-1)\big]} \\[2mm]
&\displaystyle{=\sqrt{m}\,\big[\mathcal{P}(b_j=0)-(1-\mathcal{P}(b_j= 0))\big]} \\[2mm]
&\displaystyle{=2\sqrt{m}\mathcal{P}(b_j=0)-\sqrt{m}}.
\end{array}
\]
Then,
\[\displaystyle{
\mathcal{P}(b_j=0)=\frac{1}{2}+\frac{1}{2\sqrt{m}}\mathcal{E}_{j}
}\]
\[\displaystyle{
\mathcal{P}(b_j=1)=1-\mathcal{P}(b_j=0)=\frac{1}{2}-\frac{1}{2\sqrt{m}}\mathcal{E}_{j}.
}\]

Similarly, when $q^{(j)}\neq q^{(k)}$,
\[
\begin{array}{ll}
\displaystyle{\mathcal{E}_{jk}}
&\displaystyle{:= {\rm tr}(P_{\langle j\rangle}P_{\langle k\rangle}\rho)} \\[2mm]
&\displaystyle{ = m\,\mathbb{E}_{{\bf b}\sim\mathcal{P}_{\rm f}({\bf b};\rho)}[(-1)^{b_j+b_k}]} \\[2mm]
&\displaystyle{=m\,\big[\mathcal{P}(b_j=b_k)\cdot(+1)+\mathcal{P}(b_j\neq b_k)\cdot(-1)\big]} \\[2mm]
&\displaystyle{=m\,\big[\mathcal{P}(b_j=b_k)-(1-\mathcal{P}(b_j= b_k))\big]} \\[2mm]
&\displaystyle{=2m\mathcal{P}(b_j=b_k)-m}.
\end{array}
\]
Then,
\[\displaystyle{
\mathcal{P}(b_j=b_k)=\frac{1}{2}+\frac{1}{2m}\mathcal{E}_{jk}
}\]
\[\displaystyle{
\mathcal{P}(b_j\neq b_k)=1-\mathcal{P}(b_j=b_k)=\frac{1}{2}-\frac{1}{2m}\mathcal{E}_{jk}.
}\]
However, when $q^{(j)}= q^{(k)}$
\[
\begin{array}{ll}
\displaystyle{\mathcal{P}(b_j=b_k)}
&\displaystyle{= \mathcal{P}(b_j=0)\cdot \mathcal{P}(b_k=0)+\mathcal{P}(b_j=1)\cdot \mathcal{P}(b_k=1)}\\[2mm]
&\displaystyle{= \frac{1}{2}+\frac{1}{2m}\mathcal{E}_j\mathcal{E}_k}
\end{array}
\]
and
\[
\begin{array}{ll}
\displaystyle{\mathcal{P}(b_j\neq b_k)}
&\displaystyle{= 1-\mathcal{P}(b_j=b_k)}\\[2mm]
&\displaystyle{= \frac{1}{2}-\frac{1}{2m}\mathcal{E}_j\mathcal{E}_k}.
\end{array}
\]
\qed
\end{proof}

\section{Matrix Product State}
\label{sec:mps}

Let $\Psi\in\mathbb{C}^{2^n}$ be $n$-qubit quantum state.
Generally, $\Psi$ can be represented as the matrix multiplication of $n$ matrices as 
\begin{equation}
\label{eq:mps}
\Psi
=A^{[n-1]}A^{[n-2]}\cdots A^{[1]}A^{[0]},
\end{equation}
where 
\begin{equation}
\label{eq:}
A^{[i]}\in\mathbb{C}^{\min(2^{i+1},2^{n-i-1},\chi)\times2\times \min(2^i,2^{n-i-1},\chi)}
\end{equation}
and
\begin{equation}
\label{eq:}
(A^{[i]}A^{[j]})_{k(ab)m} = \sum^{\min(2^i,2^{n-i-1},\chi)}_{l=1} (A^{[i]})_{kal}(A^{[j]})_{lbm}.
\end{equation}
The subscript in the parentheses, $(ab)$, summarizes two indices in one index, i.e., $(a,b)=(1,1),(1,2),\ldots,(2,1),(2,2),\ldots$.
The expression of Eq.~\eqref{eq:mps} is called the MPS~\cite{schollwock2011density}.
The dimension between matrices is called the bond dimension and can be restricted to $\chi$ at the expense of representation power.
The number of components of an $n$-qubit state is $2^n$, whereas the total number of matrix components of MPS is $\sim 2n\chi^2$.
For example, a $4$-qubit state $|\psi\rangle \in \mathbb{C}^{2^4}$ with bond dimension $\chi=2$ can be represented as
\[\displaystyle{
\psi_{(abcd)}
=\sum_{\alpha,\beta,\gamma,\delta,\epsilon}
\begin{bmatrix}
\begin{bmatrix}
\bullet\\ \bullet
\end{bmatrix}_{a}
&\begin{bmatrix}
\bullet\\ \bullet
\end{bmatrix}_{a}
\end{bmatrix}_{\alpha\beta}
\begin{bmatrix}
\begin{bmatrix}
\bullet\\ \bullet
\end{bmatrix}_{b}
&
\begin{bmatrix}
\bullet\\ \bullet
\end{bmatrix}_{b}
&
\begin{bmatrix}
\bullet\\ \bullet
\end{bmatrix}_{b}
&
\begin{bmatrix}
\bullet\\ \bullet
\end{bmatrix}_{b}\\[5mm]
\begin{bmatrix}
\bullet\\ \bullet
\end{bmatrix}_{b}
&
\begin{bmatrix}
\bullet\\ \bullet
\end{bmatrix}_{b}
&
\begin{bmatrix}
\bullet\\ \bullet
\end{bmatrix}_{b}
&
\begin{bmatrix}
\bullet\\ \bullet
\end{bmatrix}_{b}
\end{bmatrix}_{\beta\gamma}
\begin{bmatrix}
\begin{bmatrix}
\bullet\\ \bullet
\end{bmatrix}_{c}
&
\begin{bmatrix}
\bullet\\ \bullet
\end{bmatrix}_{c}\\[5mm]
\begin{bmatrix}
\bullet\\ \bullet
\end{bmatrix}_{c}
&
\begin{bmatrix}
\bullet\\ \bullet
\end{bmatrix}_{c}\\[5mm]
\begin{bmatrix}
\bullet\\ \bullet
\end{bmatrix}_{c}
&
\begin{bmatrix}
\bullet\\ \bullet
\end{bmatrix}_{c}\\[5mm]
\begin{bmatrix}
\bullet\\ \bullet
\end{bmatrix}_{c}
&
\begin{bmatrix}
\bullet\\ \bullet
\end{bmatrix}_{c}
\end{bmatrix}_{\gamma\delta}
\begin{bmatrix}
\begin{bmatrix}
\bullet\\ \bullet
\end{bmatrix}_{d}
\\[5mm]
\begin{bmatrix}
\bullet\\ \bullet
\end{bmatrix}_{d}
\end{bmatrix}_{\delta\epsilon}
}\]
\begin{equation}
\label{eq:}
\approx\sum_{\alpha,\beta,\gamma,\delta,\epsilon}
\begin{bmatrix}
\begin{bmatrix}
\bullet\\ \bullet
\end{bmatrix}_{a}
&\begin{bmatrix}
\bullet\\ \bullet
\end{bmatrix}_{a}
\end{bmatrix}_{\alpha\beta}
\begin{bmatrix}
\begin{bmatrix}
\bullet\\ \bullet
\end{bmatrix}_{b}
&\begin{bmatrix}
\bullet\\ \bullet
\end{bmatrix}_{b}\\[5mm]
\begin{bmatrix}
\bullet\\ \bullet
\end{bmatrix}_{b}
&\begin{bmatrix}
\bullet\\ \bullet
\end{bmatrix}_{b}
\end{bmatrix}_{\beta\gamma}
\begin{bmatrix}
\begin{bmatrix}
\bullet\\ \bullet
\end{bmatrix}_{c}
&\begin{bmatrix}
\bullet\\ \bullet
\end{bmatrix}_{c}\\[5mm]
\begin{bmatrix}
\bullet\\ \bullet
\end{bmatrix}_{c}
&
\begin{bmatrix}
\bullet\\ \bullet
\end{bmatrix}_{c}
\end{bmatrix}_{\gamma\delta}
\begin{bmatrix}
\begin{bmatrix}
\bullet\\ \bullet
\end{bmatrix}_{d}
\\[5mm]
\begin{bmatrix}
\bullet\\ \bullet
\end{bmatrix}_{d}
\end{bmatrix}_{\delta\epsilon},
\end{equation}
where $\bullet$ indicates the matrix element and the subscripts represent the index symbols of the inner submatrix elements.

Let $P$ be the identity or Pauli matrix and
\begin{equation}
\label{eq:ham}
H=\sum^{K}_{k=1} C_k \bigotimes^{n-1}_{l=0} P_{kl}
\end{equation}
be a $n$-qubit Hamiltonian that consists of a weighted sum of $K$ Pauli products.
Similar to the MPS, Eq.~\eqref{eq:ham} can be represented as matrix multiplication of $n$ matrices as
\begin{equation}
\label{eq:mpo}
H=B^{[n-1]}B^{[n-2]}\cdots B^{[1]}B^{[0]},
\end{equation}
where 
\begin{equation}
\label{eq:}
B^{[i]}\in\left\{
    \begin{array}{ll}
        \displaystyle{ \mathbb{C}^{K\times2\times2\times 1} } & i=0  \\[5mm]
        \displaystyle{ \mathbb{C}^{1\times2\times2\times K} } & i=n-1  \\[5mm]
        \displaystyle{\mathbb{C}^{K\times2\times2\times K}} & {\rm else}
    \end{array}
\right. ,
\end{equation}
\begin{equation}
\label{eq:}
    \left\{
        \begin{array}{ll}
            \displaystyle{ (B^{[i]})_{kaa'1}=C_k(P_{ki})_{aa'} } & i=0  \\[5mm]
            \displaystyle{ (B^{[i]})_{1aa'k}=(P_{ki})_{aa'} } & i=n-1  \\[5mm]
            \displaystyle{(B^{[i]})_{kaa'k}=(P_{ki})_{aa'}} & {\rm else}
        \end{array}
    \right.
\end{equation}
and
\begin{equation}
\label{eq:}
(B^{[i]}B^{[j]})_{l(ab)(a'b')m} = \sum^{K}_{k=1} (B^{[i]})_{laa'k}(B^{[j]})_{kbb'm}.
\end{equation}
For example,
\[\displaystyle{
H_{(abcd)(a'b'c'd')}=(C_1IXYI+C_2IZII+C_3IYIX)_{(abcd)(a'b'c'd')}
}\]
\begin{equation}
\label{eq:}
=\sum_{\alpha,\beta,\gamma,\delta,\epsilon}
\begin{bmatrix}
[I]_{aa'}&[I]_{aa'}&[I]_{aa'}
\end{bmatrix}_{\alpha\beta}
\begin{bmatrix}
[X]_{bb'}\\ &[Z]_{bb'}\\ &&[Y]_{bb'}
\end{bmatrix}_{\beta\gamma}
\begin{bmatrix}
[Y]_{cc'}\\ &[I]_{cc'}\\ &&[I]_{cc'}
\end{bmatrix}_{\gamma\delta}
\begin{bmatrix}
[C_1I]_{dd'} \\ [C_2I]_{dd'} \\ [C_3X]_{dd'}
\end{bmatrix}_{\delta\epsilon}
\end{equation}
The expression of Eq.~\eqref{eq:mpo} is called the matrix product operator (MPO)~\cite{schollwock2011density}.

The MPO is a set of sparse matrices and thus can be efficiently stored in classical memory in COOrdinate (COO), Compressed Sparse Row (CSR), or Compressed Sparse Column (CSC) formats.

Using the MPS and MPO, the expectation value $\langle\Psi|H|\Psi\rangle$ can be efficiently calculated by the matrix multiplication of $\{A^{[i]}\}$ and $\{B^{[i]}\}$.

All edge energy can be simultaneously evaluated by modifying the MPO of the problem Hamiltonian.
\begin{equation}
\label{eq:}
    \left\{
        \begin{array}{ll}
            \displaystyle{ (B^{[i]})_{kaa'1}=(P_{ki})_{aa'} } & i=0  \\[5mm]
            \displaystyle{(B^{[i]})_{kaa'k}=(P_{ki})_{aa'}} & {\rm else}
        \end{array}
    \right.
\end{equation}
For example,
\[\displaystyle{
\begin{bmatrix}
C_1IXYI_{(abcd)(a'b'c'd')}\\
C_2IZII_{(abcd)(a'b'c'd')}\\
C_3IYIX_{(abcd)(a'b'c'd')}
\end{bmatrix}
}\]
\begin{equation}
\label{eq:mpo2}
=\sum_{\alpha,\beta,\gamma,\delta,\epsilon}
\begin{bmatrix}
[I]_{aa'}\\&[I]_{aa'}\\&&[I]_{aa'}
\end{bmatrix}_{\alpha\beta}
\begin{bmatrix}
[X]_{bb'}\\ &[Z]_{bb'}\\ &&[Y]_{bb'}
\end{bmatrix}_{\beta\gamma}
\begin{bmatrix}
[Y]_{cc}\\ &[I]_{cc'}\\ &&[I]_{cc'}
\end{bmatrix}_{\gamma\delta}
\begin{bmatrix}
[C_1I]_{dd'} \\ [C_2I]_{dd'} \\ [C_3X]_{dd'}.
\end{bmatrix}_{\delta\epsilon}
\end{equation}

\section{Implementation Details for GW, CirCut, the RQAOA, and QRAO}
\label{sec:implement-details}
For GW, an official MATLAB implementation of SDPT3~\cite{toh1999sdpt3,tutuncu2003solving} is used to solve the SDP.
The gap tolerance is set to $10^{-9}$ for high-precision calculation.
For hyperplane cutting, which is a post-processing step in GW performed after computing the SDP solution, in-house Python code is used.
The number of hyperplane cuttings is set to $10,000$.

For CirCut~\cite{burer2002rank}, the official Fortran implementation (\url{https://www.cmor-faculty.rice.edu/~zhang/circut/index.html}) is used.
The hyperparameters for CirCut are not changed from the default ones.

For RQAOA, in-house Fortran code is used.
The analytic form of $\langle\psi_{\rm QAOA}({\boldsymbol\theta})|Z_jZ_k|\psi_{\rm QAOA}({\boldsymbol\theta})\rangle$~\cite{bravyi2020obstacles} is used to evaluate the expectation value of the problem Hamiltonian and edge energy, where $|\psi_{\rm QAOA}({\boldsymbol\theta})\rangle=e^{i\beta B}e^{i\gamma C}|+\rangle^{\otimes n}$ is a level-1 QAOA ansatz, $B=\sum^{n-1}_{i=0}X_i$, $C=\sum_{(j,k)\in E}w_{jk}Z_jZ_k$, and ${\boldsymbol\theta}=\{\beta,\gamma\}$.
Only the optimal $\gamma$ is determined using a grid search; the optimal $\beta$ is calculated for each $\gamma$ using the same method used in \cite{bravyi2022hybrid}.
We derived the concrete representation of the optimal $\beta$ as
\begin{equation}
    \beta^* = -\frac{1}{4}{\rm arctan}\left(\frac{\frac{1}{2}(R_{11}-R_{22}-R_{33}+R_{44})+R_{14}-R_{23}}{-I_{12}-I_{13}+I_{24}+I_{34}}\right)+\frac{\pi}{8},
\end{equation}
where
\begin{equation}
\begin{bmatrix}
    R_{11} & R_{12}+iI_{12} & R_{13}+iI_{13} & R_{14}+iI_{14} \\
    R_{12}-iI_{12} & R_{22} & R_{23}+iI_{23} & R_{24}+iI_{24} \\
    R_{13}-iI_{13} & R_{23}-iI_{23} & R_{33} & R_{34}+iI_{34} \\
    R_{14}-iI_{14} & R_{24}-iI_{24} & R_{34}-iI_{34} & R_{44}
\end{bmatrix}
:=\sum_{(u,v)\in E}w_{uv}\rho_{uv}
\end{equation}
and
\begin{equation}
\rho_{uv}:=
{\rm tr}_{uv}\Big(e^{i\gamma C}|+\rangle^{\otimes n}\langle+|^{\otimes n}e^{-i\gamma C}\Big).
\end{equation}
Here, ${\rm tr}_{uv}(\sigma)$ denotes the reduced state obtained by tracing out, excepting the $u$th and $v$th qubits from $\sigma$.
For an algorithm to calculate $\rho_{uv}$, see \cite{bravyi2022hybrid}.
The parameter spaces are set to $\beta\in[0,\frac{\pi}{2}]$ and $\gamma\in[0,\pi]$.
Once the optimal $\gamma$ is found in $50$ equidistant grids of $[0,\pi]$, the more precise solution is searched for in the fine-grained grids around the course-grained solution with a mesh size is $1/50$.
Fine-graining is repeated twice, resulting in a mesh size of $1/250$ of the parameter space.

For QRAO, in-house Python code is used.
The MPS ansatz with bond dimension $\chi=2$ is used for $|\psi({\boldsymbol\theta})\rangle$ and Pauli rounding is used for decoding.
The same optimizer as used for RQRAO, L-BFGS with the strong Wolfe condition, is used to optimize $|\psi({\boldsymbol\theta})\rangle$.

\begin{figure*}[t]
    \centering
    \subfigure[]{
        \centering
        \includegraphics[scale=.8]{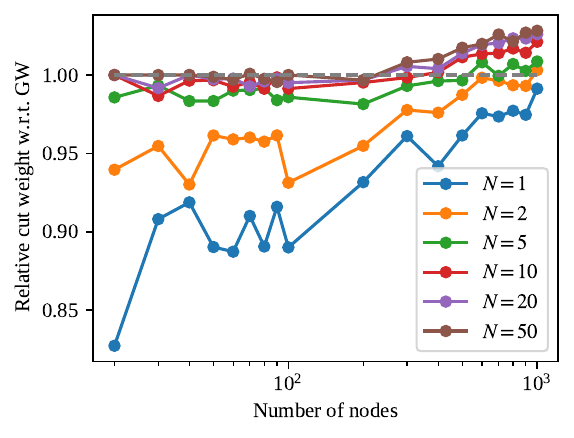}
        \includegraphics[scale=.8]{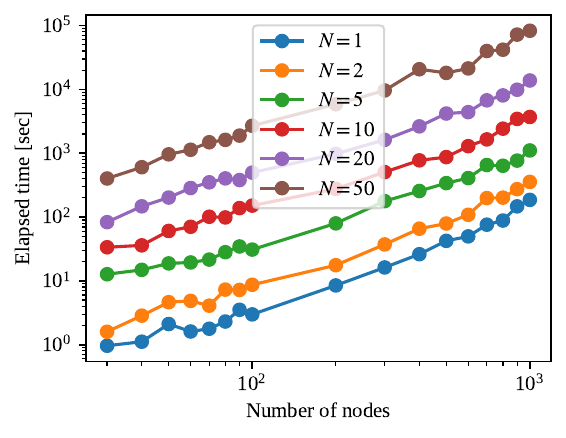}
        \label{fig:ablation-N}
    }
    \subfigure[]{
        \centering
        \includegraphics[scale=.8]{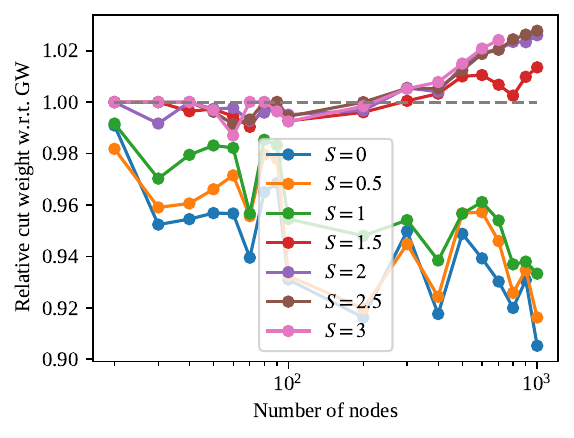}
        \includegraphics[scale=.8]{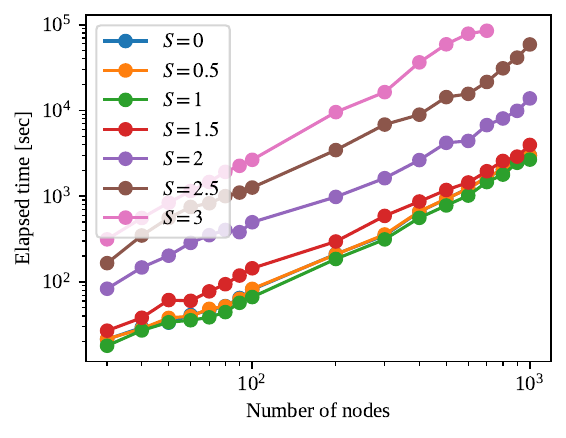}
        \label{fig:ablation-S}
    }
    \subfigure[]{
        \centering
        \includegraphics[scale=.8]{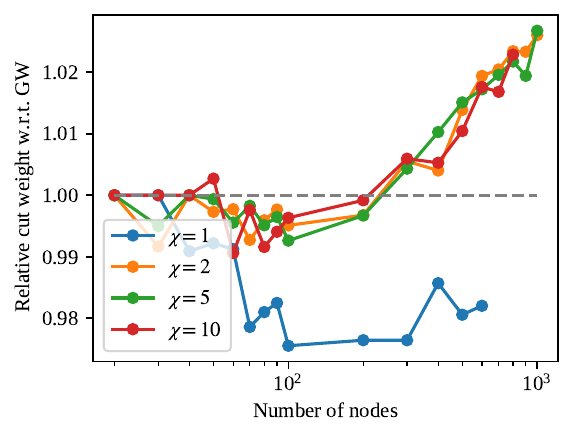}
        \includegraphics[scale=.8]{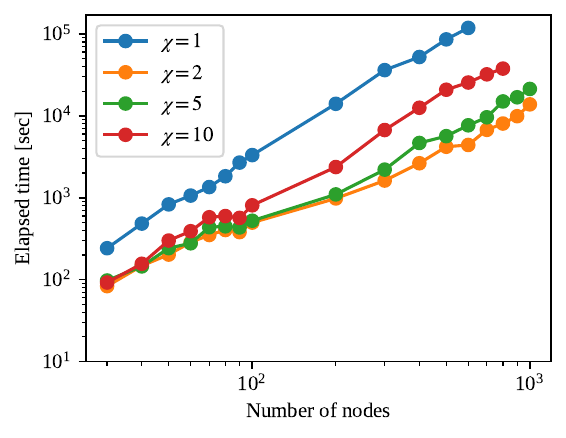}
        \label{fig:ablation-chi}
    }
    \caption{Results of the ablation study. Only one hyperparameter is varied for each trial. The problem graph instances are $3$-regular graphs with $\pm1$ edge weights. For each number of nodes, $10$ graph instances are randomly generated. The varying hyperparameters are the
    (a) number of ensembles $N$;
    (b) scale factor $S$; and
    (c) bond dimension $\chi$.
    }
  \label{fig:ablation-others}
\end{figure*}

\begin{figure*}[t]
    \centering
    \subfigure[]{
        \centering
        \includegraphics[scale=1]{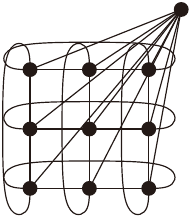}
        \label{fig:2d-toric-a}
    }
    \subfigure[]{
        \centering
        \includegraphics[scale=.8]{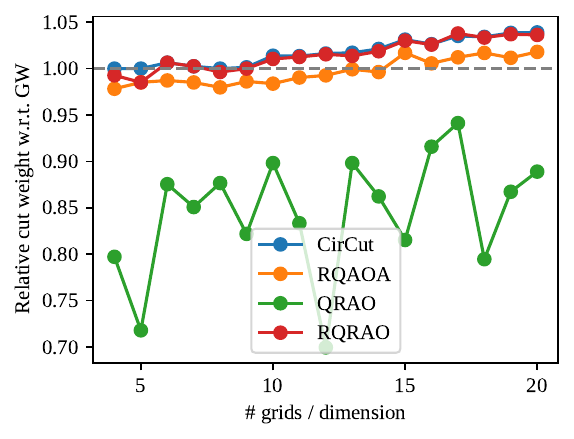}
        \includegraphics[scale=.8]{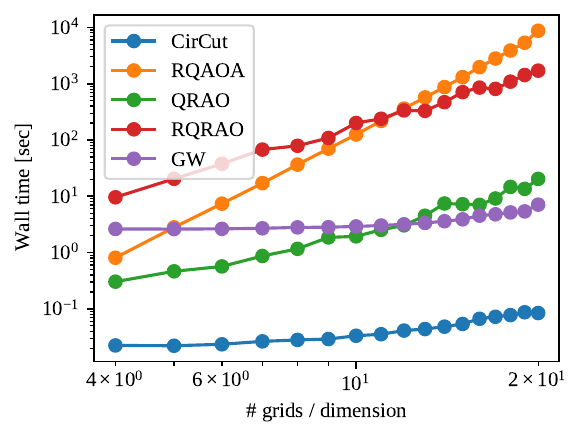}
        \label{fig:2d-toric-b}
    }
    \caption{\red{(a) Example of 2D toric graph with one external nodes of size $3\times 3$ (\# girds / dimension $=3$); (b) Summary of results for 2D-toric graphs.}
    }
  \label{fig:2d-toric}
\end{figure*}

\section{Additional Sensitivity Analysis Results (\texttt{ExpHyp})}
\label{sec:ablation-others}

Here, the sensitivities of all RQRAO hyperparameters except for the number of embeddings per qubit $m$ on the performance of the cut weight and runtime are reported.
The results are summarized in Fig.~\ref{fig:ablation-others}.
In brief, a higher cut weight is obtained through a longer runtime.

Except for the number of embeddings per qubit $m$, the sensitive hyperparameters with respect to cut weight performance are the number of ensembles $N$ and scale factor $S$.
In Fig.~\ref{fig:ablation-N}, the cut weight increases as the number of ensembles $N$ increases.
In addition, as the number of ensembles increases, the number of edges whose parities are simultaneously determined in one optimization step tends to decrease.
This is because the edges, which are a candidate for determining the parity when the number of ensembles is small, drop out as the accuracy of the parity confidence improves, which is achieved by increasing the number of ensembles.
This trend (i.e., the number of candidates decreases with as the accuracy of the parity confidence increases), can be seen in the larger scale factor $S$ in Fig.~\ref{fig:ablation-S}.
A larger $S$ tends to decrease the number of candidates because the number of edges with nonzero $\mathcal{E}_{jk}$ decreases.

The effect of the bond dimension $\chi$ on the cut weight performance needs to be carefully considered.
Recall that a large entanglement can be treated using a large bond dimension.
The $\chi=1$ result is worse both in terms of relative cut weight and runtime than the $\chi>1$ results.
In the case of $\chi=1$, $|\psi({\boldsymbol\theta})\rangle$ is restricted to a direct product state that tends to be trapped in a lower energy state than in the $\chi>1$ case because the search space is much smaller.
As a result, each edge energy in an ensemble, $\mathcal{E}^{(t)}_{jk}$, is no longer biasedly distributed in either positive or negative sides because the lower energy state includes smaller information of a large cut weight.
By contrast, when $\chi>1$, no significant difference in the relative cut weight can be seen.
This does not mean, however, that a large entanglement will not affect the cut weight performance because even the value of $\chi=10$, which was the maximum value considered in the experiment, is small compared with the maximum bond dimension that the system can take, which is $\chi_{\rm max}=2^{\lfloor n/2 \rfloor}$.
For example, consider the case of $|V|=300$.
The number of qubits is assumed to be $300/3=100$, which can be achieved by a $(3,1)$-QRAC formulation if the graph is sufficiently sparse.
In this case, $\chi_{\rm max}= 2^{\lfloor 99/2 \rfloor}\approx 10^{14}$, which is much larger than the experiments considered in this study.
To take into account a large bond dimension within a reasonable runtime, computation using a real quantum device would be required.

\red{
\section{Results for 2D Toric Graphs}
\label{sec:2d-toric}
While the specific characteristics of instances where RQAOA performs particularly well are not fully understood, in~\cite{bravyi2020obstacles} RQAOA was compared against GW using 2D toric graphs with one extra node connected with all other nodes (Fig.~\ref{fig:2d-toric-a}), referred to as $G_{\rm toric}$.
To verify that our proposed model also surpasses RQAOA on $G_{\rm toric}$, we conducted numerical simulations for $G_{\rm toric}$ under the same settings as \texttt{ExpScale}.
In Fig.~\ref{fig:2d-toric-b}, it is evident that RQAOA outperforms GW in terms of cut weight when the number of grids per dimension exceeds 15, consistent with the findings of \cite{bravyi2020obstacles}. However, both CirCut and our proposed RQRAO surpass RQAOA.
Note that the trend in time complexity is nearly identical to the results observed for 3-regular graphs, as shown in Fig.~\ref{fig:scalability-time}.
}

\FloatBarrier

\section{Related Work}
\label{sec:related}
\subsection{Classical Algorithm for the \maxcut problem}
\label{sec:classical-maxcut}

Before reviewing the classical algorithms for the \maxcut problem, we reformulate Definition~\ref{def:maxcut} to make the existing classical algorithms easier to understand.
Because $(-1)^{b_j+b_k}=v_jv_k$ where $v_j:=(-1)^{b_j}$, the \maxcut problem can be reformulated as
\begin{equation}
\label{eq:def-maxcut-re}
\max_{{\bf X}} \texttt{CW}_{\rm X}({\bf X}),
\end{equation}
where
\begin{equation}
\label{eq:def-cut-re}
\texttt{CW}_{\rm X}({\bf X}):= {\rm tr}({\bf w}{\bf F}({\bf X})),
\end{equation}
${\bf F}({\bf X}):=({\bf 1}-{\bf X})/2$, and ${\bf X}:={\bf v}{\bf v}^{\mathsf T}$.
By definition, ${\bf X}$ is restricted as ${\rm diag}({\bf X})={\bf 1}$, ${\rm rank}({\bf X})=1$, and ${\bf X}\succeq0$.

The best known theoretically guaranteed generic algorithm for the \maxcut problem is the GW~\cite{goemans1995improved} as stated in Sec.~\ref{sec:introduction}.
In the GW algorithm, the constraint ${\rm rank}{(\bf X})=1$ is removed.
This transforms the \maxcut problem to a semidefinite program (SDP), which can be solved in polynomial time for the number of nodes.
After solving the SDP, the solution ${\bf v}\in\{\pm1\}^{|V|}$ needs to be extracted from the obtained ${\bf X}$.
When ${\bf X}$ is written in the form ${\bf X}={\bf vv}^\mathsf{T}$, the solution is ${\bf v}$.
If not, as is generally the case here, the obtained ${\bf X}$ is decomposed as ${\bf X}={\bf L}{\bf L}^\mathsf{T}$ using, for example, Cholesky decomposition, where ${\bf L}=\begin{bmatrix}{\bf l}_1&\dots&{\bf l}_{|V|}\end{bmatrix}^\mathsf{T}\in\mathbb{R}^{|V|\times|V|}$ is a lower triangular matrix and ${\bf l}_i\in{\mathbb R}^{|V|}$ is a unit vector.
Then, the solution $\{v_i\}$ is determined by the random projection of ${\bf l}_i$ as $v_i=1$ if ${\bf l}_i^\mathsf{T}{\bf r}\ge0$ and $v_i=-1$ otherwise.
Here, ${\bf r}\in{\mathbb R}^{|V|}$ is a random vector uniformly distributed on the $|V|$-dimensional unit sphere.
This can be viewed as dividing the set $\{{\bf l}_i\}^{|V|}_{i=1}$ into two sets by a hyperplane with ${\bf r}$ as the normal vector.
When the hyperplane cutting is performed uniformly at random, the lower bound of the expected cut weight is $0.878\,\texttt{CW}({\bf b}^*)$~\cite{goemans1995improved}.
This bound is explained as follows.
According to \cite[Theorem 3.1]{goemans1995improved}, the \maxcut problem can be further reformulated as 
\begin{equation}
\{{\bf l}^{\rm opt}_i\}:=\mathop{\rm argmax}_{\{{\bf l}_i\}}{\mathbb E}_{{\bf b}\sim\mathcal{P}^{\rm MC}({\bf b};\{{\bf l}_i\})}[\texttt{CW}({\bf b})]    
\end{equation}
where
\begin{equation}
    \mathcal{P}^{\rm MC}(b_j\neq b_k;\{{\bf l}_i\}):=\frac{{\rm arccos}({\bf l}_j^\mathsf{T}{\bf l}_k)}{\pi}.
\end{equation}
In the GW algorithm, the objective is rewritten as
\begin{equation}
\label{eq:gw-re}
    \{{\bf l}^*_i\}:=\mathop{\rm argmax}_{\{{\bf l}_i\}}{\mathbb E}_{{\bf b}\sim\mathcal{P}^{\rm GW}({\bf b};\{{\bf l}_i\})}[\texttt{CW}({\bf b})]
\end{equation}
where
\[\displaystyle{
    \mathcal{P}^{\rm GW}(b_j\neq b_k;\{{\bf l}_i\})
    :=\frac{1-{\bf l}_j^\mathsf{T}{\bf l}_k}{2}
}\]
\begin{equation}
\label{eq:ineq}
    \le \frac{1}{0.878} \mathcal{P}^{\rm MC}(b_j\neq b_k;\{{\bf l}_i\}).
\end{equation}
The last inequality of Eq.~\eqref{eq:ineq} comes from \cite[Lemma 3.4, Lemma 3.5]{goemans1995improved}.
\footnote{
$\min_{x\in[-1,1]}\frac{{\rm arccos}(x)}{\pi}\big(\frac{1-x}{2}\big)^{-1}\approx0.878$.
}
The right-hand side of Eq.~\eqref{eq:gw-re} is an SDP, and hence it can be solved in polynomial time.
The resulting ${\mathbb E}_{{\bf b}\sim\mathcal{P}^{\rm MC}({\bf b};\{{\bf l}^*_i\})}[\texttt{CW}({\bf b})]$, which is the expected cut weight of random hyperplane cutting, is guaranteed to be greater than $0.878\,{\mathbb E}_{{\bf b}\sim\mathcal{P}^{\rm MC}({\bf b};\{{\bf l}^{\rm opt}_i\})}[\texttt{CW}({\bf b})]=0.878\max_{{\bf b}}\texttt{CW}({\bf b})$ because
\[\displaystyle{
    \mathcal{P}^{\rm MC}(b_j\neq b_k;\{{\bf l}^*_i\})
    \ge 0.878\,\mathcal{P}^{\rm GW}(b_j\neq b_k;\{{\bf l}^*_i\})
}\]
\[\displaystyle{
    \ge 0.878\,\mathcal{P}^{\rm GW}(b_j\neq b_k;\{{\bf l}^{\rm opt}_i\})
}\]
\begin{equation}
    = 0.878\,\mathcal{P}^{\rm MC}(b_j\neq b_k;\{{\bf l}^{\rm opt}_i\}).
\end{equation}
Equation~\eqref{eq:gw-re} is very similar to Eq.~\eqref{eq:prop2}, but the rounding process of the GW is not the same as that of Eq.~\eqref{eq:prop1}.
Instead of carrying out ${\bf b}^{\rm R}=\mathop{\rm argmax}_{{\bf b}}\mathcal{P}^{\rm GW}({\bf b};\{{\bf l}^*_i\})$ or ${\bf b}^{\rm R}=\mathop{\rm argmax}_{{\bf b}}\mathcal{P}^{\rm MC}({\bf b};\{{\bf l}^*_i\})$, a number of ${\bf b}$ obeying $\mathcal{P}^{\rm MC}({\bf b};\{{\bf l}^*_i\})$ are sampled by random hyperplane cutting and then, the best $\texttt{CW}({\bf b})$ is adopted as the candidate solution.

To solve the GW algorithm faster with less memory, various SDP solvers have been developed~\cite{majumdar2020recent}.
The runtime of the state-of-the-art SDP solvers are $\widetilde{O}(\sqrt{n}(mn^2+m^\omega+n^\omega)\log(1/\epsilon)))$~\cite{jiang2020faster} and $\widetilde{O}((\sqrt{n}(m^2+n^4)+m^\omega+n^{2\omega})\log(1/\epsilon)$~\cite{huang2022solving}, where $n$ is the problem size, $m$ the number of constraints, $\omega$ the matrix multiplication constant $\le2.372927$, and $\epsilon$ the relative accuracy.
In the case of the \maxcut problem, the runtime is $\widetilde{O}(|V|^{3.5}\log(1/\epsilon))$ because $n=m=|V|$.
The overall time complexity for the GW algorithm depends on the runtime of SDP because the na\"ive Cholesky decomposition takes $O(n^3)=O(|V|^3)$ to run, which is faster than the state-of-the-art SDP solver, as long as the number of trials of random hyperplane cutting is less than the time complexity of the SDP.
Note that the computational complexity of SDP with respect to the number of nodes can be further reduced to $\widetilde{O}(|V|/\epsilon^{3.5})$ using a specific algorithm for the \maxcut, by taking a disadvantage with respect to the accuracy~\cite{lee2020widetilde}.

Although there is no approximation ratio guarantee nor computational complexity bound, there are several classical heuristics that are empirically faster and have higher approximation ratios than the GW algorithm.
The rank-two relaxation algorithm~\cite{burer2002rank} relaxed ${\rm rank}({\bf X})=1$ to ${\bf X}={\bf AA}^\mathsf{T}$, where ${\bf A}=\begin{bmatrix}{\bf a}_1&\dots&{\bf a}_{|V|}\end{bmatrix}^\mathsf{T}$, ${\bf a}_j=\begin{bmatrix}\cos\theta_j&\sin\theta_j\end{bmatrix}^\mathsf{T}$, and $\theta_j\in[0,2\pi)$.
Similar to the GW, the objective of the rank-two relaxation can be reformulated as
\begin{equation}
\{\theta^*_j\} := \mathop{\rm argmax}_{\{\theta_j\}}\mathbb{E}_{\mathcal{P}^{\rm R2}({\bf b};\{\theta_j\})}[\texttt{CW}({\bf b})],
\end{equation}
where
\begin{equation}
\mathcal{P}^{\rm R2}(b_j\neq b_k;\{\theta_j\})=\sin^2\Big(\frac{\theta_j-\theta_k}{2}\Big).
\end{equation}
As with the GW algorithm, the rank-two relaxation does not take the sample with the highest probability as a candidate.
Moreover, because the hyperplane is only two-dimensional, random hyperplane cutting is not used.
Instead, exhaustive search is used, expressed as
\begin{equation}
    {\bf b}^{\rm R2}=\texttt{CW}({\bf b}^*(\alpha^*)),
\end{equation}
where
\begin{equation}
    b^*_j(\alpha)=\left\{
    \begin{array}{lll}
        0 & {\rm for} & \cos(\theta^*_j+\alpha) \ge 0 \\[2mm]
        1 & {\rm for} & \cos(\theta^*_j+\alpha) < 0 \\
    \end{array}
    \right.
\end{equation}
and
\begin{equation}
    \alpha^*=\mathop{{\rm argmax}}_{\alpha}\texttt{CW}({\bf b}^*(\alpha)).
\end{equation}
To optimize ${\theta_j}$, a gradient-based optimizer is used.
In addition, $\{1,2\}$-local search and restarting optimization are incorporated.
Here, $k$-local search is an exhaustive method that searches the entire solution within Hamming distance $k$ of the current solution.
Breakout local search~\cite{benlic2013breakout} repeats the 1-local search and perturbation of the current solution using a tabu list.
Once a local search falls into a local optimum, perturbation is added to the current solution, and then the local search is restarted.
In the perturbation phase, one of three kinds of perturbation is applied with a specific probability.
The tabu list is used to prohibit the perturbation of the listed nodes so that the same local minimum will not be visited again.

\subsection{Quantum Algorithms for the \maxcut problem}
\label{sec:quantum-maxcut}

The quantum algorithms for the \maxcut problem are categorized into two research streams.
One is quantum approximate optimization algorithms (QAOAs)~\cite{qaoa1,wang2018quantum} and the other is quantum algorithms for SDP~\cite{brandao2017quantum,van2020quantum,brando_et_al:LIPIcs:2019:10603,vanapeldoorn_et_al:LIPIcs:2019:10675,kerenidis2020quantum,brandao2022faster,bharti2022noisy,patel2021variational,patti2023quantum}.
Here, quantum annealing is excluded.

QAOA variants are well summarized in a recent review~\cite{blekos2023review} and are not re-summarized here.

The quantum SDP algorithms are categorized into two types.
One encodes the solution to a density matrix and the other uses QRAM~\cite{giovannetti2008quantum} to store the solution.

For the density matrix encoding approach, there are two types of encoding methods.
One uses a Gibbs state to represent a solution~\cite{brandao2017quantum,van2020quantum,brando_et_al:LIPIcs:2019:10603,vanapeldoorn_et_al:LIPIcs:2019:10675,brandao2022faster}, which can be regarded as a solution that is encoded in the corresponding Hamiltonian.
The Hamiltonian used in the Gibbs state is updated by the matrix multiplicative weight update method~\cite{arora2007combinatorial} or the matrix exponentiated gradient update method~\cite{tsuda2005matrix}.
The other density matrix encoding approach uses a parametrized quantum circuit to represent a solution~\cite{bharti2022noisy,patel2021variational,patti2023quantum}.
The parameter is optimized by frameworks using the variational quantum algorithm~\cite{vqa,vqa2}.

The state-of-the-art quantum SDP algorithm~\cite{huang2022faster} is based on the primal-dual central path method, which uses quantum linear algebra and stores the solution in QRAM.
The runtime of this quantum SDP algorithm is $O((mn^{1.5}+n^3){\rm poly}(\kappa,\log(mn/\epsilon))$ where $\kappa$ is the condition number of the matrices appearing in the algorithm, and $\epsilon$ is the accuracy parameter.
In the case of the \maxcut problem, we have $\widetilde{O}(|V|^3{\rm poly}(\kappa,\log(|V|^2/\epsilon))$, which is a slight improvement on its classical counterpart, which is $\widetilde{O}(|V|^{3.5}\log(1/\epsilon))$.

Although these algorithms have the potential to solve an SDP faster than classical ones, the approximation ratio is the same as that of the GW algorithm, which is empirically worse than that of classical heuristics.

Very recently, \cite{munoz2023low} proposed the quantum-inspired algorithm based on the observation that ADAPT-QAOA~\cite{zhu2022adaptive} can be approximately expressed as the Clifford circuit, which can be efficiently simulatable using a classical computer.
Their model showed the comparable cut weight performance to the GW with several thousand hyperplane cuttings.
\subsection{QRACs}
\label{sec:qrac}
The concept that encodes $m$ bits into $n$ (qu)bits was originally proposed by Wiesner as {\it conjugate coding}~\cite{wiesner1983conjugate} and later rediscovered by Ambainis et al.~\cite{ambainis1999dense} as {\it (quantum) random access codes ((Q)RACs)}.
$(m,n,p)$ random access encoding is defined as the function that maps $m$ classical bits into $n$ (qu)bits with a decoding probability at least $p$~\cite[Definition 1.1]{ambainis1999dense}.
The minimum number of $n$ with $p>1/2$ is bounded as $(1-\eta(p))m$ for both classical~\cite[Theorem 2.1]{ambainis1999dense} and quantum~\cite[Theorem 2.3]{nayak1999optimal} settings, where $\eta(p)$ is the binary entropy function.
In the case of $(2,1,p)$- and $(3,1,p)$-QRACs, there exists $p=\frac{1}{2}+\frac{1}{2\sqrt{2}}$~\cite[Lemma 3.1]{ambainis1999dense} and $p=\frac{1}{2}+\frac{1}{2\sqrt{3}}$~\cite[attributed to Chuang]{ambainis1999dense}, respectively, which are consistent with our results (Eq.~\eqref{eq:prob-bj}) in the case of $\mathcal{E}_j\in\{\pm1\}$.
It was proved that $(2^n,n,>1/2)$ classical random access encoding~\cite[Theorem 5.5]{iwama2007unbounded}
and $(2^{2n},n,>1/2)$ quantum random access encoding~\cite[Theorem 2]{hayashi20064} do not exist.
From these observations, $(3,1)$-QRAC shows the maximum number of classical bits that can be embedded in one qubit with a decoding probability that is larger than $1/2$.
This is why we only consider $m\in\{1,2,3\}$ in this study.
In addition to $(m,1,p)$-QRAC, $(m,2,p)$-QRAC is particularly of interest and well studied in terms of the theoretical bound~\cite{hayashi20064,liabotro2017improved,imamichi2018constructions,manvcinska2022geometry} and concrete construction scheme~\cite{teramoto2023quantum}.
For arbitrary $(m,n,p)$-QRAC, the upper bound of decoding probability $p$ is improved by considering shared randomness~\cite{ambainis2008quantum}, shared entanglement~\cite{pawlowski2010entanglement}, and $d$-level system~\cite{tavakoli2015quantum,liabotro2017improved}.
The information that can be decoded is not limited to the $m$ bits that are encoded, but also extends to functions that use $k$-out-of-$m$ bits~\cite{ben2008hypercontractive,doriguello2021quantum}.
Using this extension, QRACs have also been applied in machine learning to efficiently encode discrete features~\cite{yano2021efficient}.


\end{document}